\newif\iflong\longtrue

\documentclass[numberwithinsect]{lipics-v2021}
\iflong \usepackage[bibliography=common]{apxproof}
\else   \usepackage[bibliography=common,appendix=strip]{apxproof}
\fi
\usepackage{soul}
\usepackage{mathtools}
\usepackage{xspace}

\title{A finite presentation \\ of graphs of treewidth at most three}
\titlerunning{A finite presentation of graphs of treewidth at most three}

\author{Amina Doumane}{Plume, LIP, CNRS, ENS de Lyon, France}{amina.doumane@ens-lyon.fr}{}{}
\author{Samuel Humeau}{Plume, LIP, CNRS, ENS de Lyon, France}{samuel.humeau@ens-lyon.fr}{}{}
\author{Damien Pous}{Plume, LIP, CNRS, ENS de Lyon, France}{damien.pous@ens-lyon.fr}{[0000-0002-1220-4399]}{}
\authorrunning{Amina Doumane, Samuel Humeau, and Damien Pous}
\Copyright{Amina Doumane, Damien Pous, and Samuel Humeau}
\funding{This work was supported by the LABEX MILYON (ANR-10-LABX-0070) of Université de Lyon, within the program ``Investissements d'Avenir'' (ANR-11-IDEX- 0007) operated by the French National Research Agency (ANR).}
\acknowledgements{The second author would like to thank Ugo Giocanti for
several discussions about graph minors and connectivity.}
\iflong
  \hideLIPIcs
  \relatedversiondetails[cite={dhp:icalp24:tw3}]
  {Published version}
  {https://doi.org/10.4230/LIPIcs.ICALP.2024.135}
\else
  \relatedversiondetails[cite={dhp:icalp24:tw3:hal}]
  {Full Version}{https://hal.science/hal-04560570}
\fi
\supplementdetails{Software}{https://perso.ens-lyon.fr/damien-pous/hypergraph/}
\nolinenumbers

  \EventEditors{Karl Bringmann, Martin Grohe, Gabriele Puppis, and Ola Svensson}
  \EventNoEds{4}
  \EventLongTitle{51st International Colloquium on Automata, Languages, and Programming (ICALP 2024)}
  \EventShortTitle{ICALP 2024}
  \EventAcronym{ICALP}
  \EventYear{2024}
  \EventDate{July 8--12, 2024}
  \EventLocation{Tallinn, Estonia}
  \EventLogo{}
  \SeriesVolume{297}
  \ArticleNo{135}

\newtheoremrep{lemma}[theorem]{Lemma}
\newtheoremrep{proposition}[theorem]{Proposition}
\newtheorem{statement}[theorem]{Statement}
\newtheorem{convention}[theorem]{Convention}

\makeatletter
\@fleqnfalse
\@mathmargin\@centering
\makeatother

\iflong
\newcommand\citeapp[2]{Appendix~\ref{#1}}
\else
\newcommand\citeapp[2]{\cite[Appendix~#2]{dhp:icalp24:tw3:hal}}
\fi

\newcommand\id{\mathrm{id}}
\newcommand\eqdef\triangleq

\newcommand\NN{\mathbb N}

\newcommand\set[1]{\left\{#1\right\}}

\newcommand\tuple[1]{\left\langle #1\right\rangle}

\newcommand\tf{{\rm f}}
\newcommand\tff{\tf\tf}
\newcommand\tl{{\rm l}}
\newcommand\tll{\tl\tl}
\newcommand\ts{{\rm s}}
\newcommand\tss{{\circ}}
\newcommand\tset{\mathcal{T}}
\newcommand\tiso{\cong}
\newcommand\giso{\simeq}
\newcommand\teq{\equiv}
\newcommand\LL{\mathcal L}
\renewcommand\top\emptyset
\newcommand\glt{\prec}%
\newcommand\sfa[1]{S(#1)}%
\newcommand\tc[1]{#1^\circ}

\newcommand\f{\tf}
\newcommand\goft{g}

\begin{document}

\maketitle

\begin{abstract}  
  We provide a finite equational presentation of graphs of treewidth
  at most three, solving an instance of an open problem by Courcelle
  and Engelfriet. 

  We use a syntax generalising series-parallel expressions, denoting
  graphs with a small interface.  We introduce appropriate notions of
  connectivity for such graphs (components, cutvertices, separation
  pairs).  We use those concepts to analyse the structure of graphs of
  treewidth at most three, showing how they can be decomposed
  recursively, first canonically into connected parallel components,
  and then non-deterministically. The main difficulty consists in
  showing that all non-deterministic choices can be related using only
  finitely many equational axioms.
\end{abstract}

\section{Introduction}

Treewidth is a graph parameter measuring how close a graph is from a
forest.  It is frequently encountered in parameterised complexity:
many NP-complete problems become polynomial or even linear once
parameterised using treewidth~\cite{DBLP:conf/mfcs/Bodlaender97}.
This parameter admits many equivalent definitions.  It was discovered
at least by Bertelè and Brioschi~\cite{DBLP:journals/jct/BerteleB73},
Halin~\cite{halin1976s}, and Robertson and
Seymour~\cite{DBLP:journals/jal/RobertsonS86} via tree decompositions
for their celebrated graph minor theorem.  It was subsequently
characterised using $k$-trees~\cite{arnborg1986characterization},
$k$-elimination
graphs~\cite{wimer1987linear,scheffler1987linear,sanders1993linear},
and chordal graphs~\cite{DBLP:journals/jal/RobertsonS86}.

One may also consider syntaxes, most often generalising
series-parallel expressions~\cite{DBLP:journals/jacm/ArnborgCPS93,
  DBLP:books/daglib/0030488, DBLP:conf/gst/Courcelle91}. There, the
idea is to use terms to denote graphs, and an important problem is to
understand when two terms denote the same graph.  To this end,
Courcelle and Engelfriet gave an equational axiomatisation for
arbitrary graphs~\cite[p. 117]{DBLP:books/daglib/0030804}: a
structured set of equations on terms from which it is possible to
equate all terms denoting a given graph.  Unfortunately, this
axiomatisation is infinite. They note how finite fragments of the
syntax make it possible to capture precisely the graphs up to a given
treewidth, while finite restrictions of their axiomatisation seem to
be incomplete.

This brings them to the following
question~\cite[p.~118]{DBLP:books/daglib/0030804}: for a given bound
$k$, is there a finite equational presentation of graphs of treewidth
at most $k$?  The case of $k=1$ concerns forests and is relatively easy.
The case of $k=2$ has been given a solution a few years
ago~\cite{DBLP:conf/mfcs/Cosme-LlopezP17, DBLP:conf/mfcs/DoczkalP18}. %
Here we give a positive answer for $k=3$: we provide a syntax of terms
with an interpretation map $\goft$ from terms to graphs whose image is
exactly the set of graphs of treewidth at most three, and we give a
finite list of equations whose equational theory $(\teq)$
characterises graph isomorphism $(\giso)$. In symbols, we prove
that for all terms $t,u$,
\begin{align*}
  g(t)\giso g(u)\quad\text{ if and only if }\quad t\teq u
\end{align*}
Like Arnborg, Courcelle, Proskurowski and
Seese~\cite{DBLP:journals/jacm/ArnborgCPS93}, we work with hypergraphs
with a list of designated vertices, the \emph{sources}, used as an
interface to perform the following operations:
\begin{itemize}
\item \emph{parallel composition} $(\parallel)$: glue the graphs
  together along their sources.
\item \emph{permutation} $(pG)$: given a permutation $p$, reorder the
  sources of $G$ according to $p$.
\item \emph{lift} $(\tl G)$: add an isolated vertex to $G$ and append
  it as last source.
\item \emph{forget} $(\tf G)$: remove the last source of $G$ (keeping
  it as a mere vertex of the graph).
\end{itemize}
\begin{figure}[t]
  \begin{align*}
    \arraycolsep=3.5pt
    \begin{array}{ccccccc}
      \includegraphics{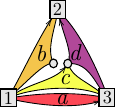}
      &\raisebox{2em}{$\giso$}&
       \includegraphics{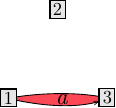}
      &\raisebox{2em}{$\parallel$}&
       \includegraphics{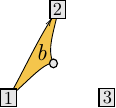}
      &\raisebox{2em}{$\parallel$}&
       \includegraphics{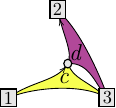}
      \\\\
      \tl\tf b \parallel \tf\big((24)\tl (c\parallel(23)\tl a)\parallel (14)\tl d\big)
      &\teq&(23)\tl a
      &\parallel&\tl\tf b
      &\parallel&\tf\big((24)\tl c\parallel (14)\tl d\big)\\
    \end{array}
  \end{align*}
  \caption{Two parsings of a given graph.}
  \label{fig:graphs}
\end{figure}
Consider the four graphs depicted in Figure~\ref{fig:graphs}, each
with three sources denoted with numbered squares. The neighbours of
each edge are ordered, which we indicate by drawing an arrow from the
first to the second neighbour. The first graph on the left is the
parallel composition of the three other ones.  The second one can be
obtained from a binary edge $a$ (with interface its endpoints) by
applying a lift to add a third isolated source and then swapping the
last two sources: $(23)\tl a$.  The third one can be obtained from a
ternary edge $b$ (again with interface its endpoints), by forgetting
the third source and adding a fresh one via a lift: $\tl\tf b$.  The
last one can be constructed as
$\tf\big((24)\tl c\parallel (14)\tl d\big)$: reasoning top-down, we
promote the inner vertex as a fourth source, and then we put in
parallel two graphs each with a single ternary edge connecting three
out of the four sources---both being obtained as appropriate
permutations of lifted edges.

We call \emph{parsings} the expressions we obtain when decomposing
graphs as above.

The way we parsed the graph on the right generalises to all graphs: first
promote all inner vertices as sources, and then build a large parallel
composition of appropriately permuted and lifted edges.  For instance,
the first graph on the left also admits the following parsing:
\begin{align*}
  \tf\tf\big((23)\tl\tl\tl a\parallel (35)\tl\tl b\parallel (24)\tl\tl c\parallel (124)\tl\tl d\big)
\end{align*}
While such a parsing exists for all graphs, it goes through an
intermediate graph with a large interface. Instead, the graphs of
treewidth at most $k$ are exactly those for which we can find a
parsing where all intermediate graphs have at most $k+1$ sources
(Proposition~\ref{prop:tw:terms}, \cite[Proposition
4.1]{DBLP:journals/jacm/ArnborgCPS93}). This syntax for denoting
precisely the graphs of treewidth at most $k$ is the starting point
for the present work.

A derived operation is of particular interest. Consider three graphs
$x$, $y$, $z$ each with three sources, and combine them as depicted on
the right to obtain a new graph $\tss(x,y,z)$. This operation can be
defined from the previous ones using the expression below.
\begin{align*}
  \raisebox{-2.3em}{\includegraphics{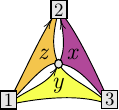}}&&
  \tss(x,y,z) \eqdef \tf\big((14)\tl x \parallel (24)\tl y \parallel (34)\tl z \big)
\end{align*}
It can be generalised into a $n$-ary operation on graphs with $n$
sources, and when $n=2$ we recover the usual notion of \emph{series
  composition} (with its arguments reversed).

Our goal is then to understand which laws are satisfied by the
previous operations. Amongst the natural ones, we have that parallel
composition is associative and commutative, that permutations commute
over parallel compositions, and that applying two permutations in a
row amounts to applying the composite permutation. There are more
involved ones. For instance, we may also parse the first graph of
Figure~\ref{fig:graphs} by keeping edges $a$ and $c$ together as long
as possible, resulting in the expression written below it. We thus
have two rather different parsings for the same graph, which our
axiomatisation should equate.

We proceed in two steps.  First we use connectivity arguments in order
to decompose any graph into a parallel composition of permuted lifts
of \emph{full prime graphs}: non-empty graphs which are connected through
paths not using sources except possibly at their
endpoints. Figure~\ref{fig:graphs} actually provided an example of such a
decomposition, which is always unique for a given graph. Using a few natural
axioms, we show that every term can be rewritten under such a form
(Proposition~\ref{prop:tm:decomp}). This makes it possible to focus on
full prime graphs in the second step, which is where the main
difficulties arise.

A key property of full prime graphs of bounded treewidth is that
either they are \emph{atomic} (i.e., reduced to a permutation of an
edge), or they have a \emph{forget point}: at least one of their inner
vertices can be promoted into a source without increasing the
treewidth, thus making it possible to parse the graph as a forget
operation.  The difficulty is that forget points are not unique,
resulting in several ways of parsing non-atomic full prime
graphs. Consider for instance the following tetrahedron on the left,
with only two vertices marked as sources.
\[
  \includegraphics{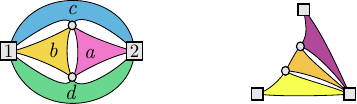}
\]
Each of the two inner vertices is a forget point, so that modulo
appropriate permutations of edges, we may parse this graph as
$\tf(\tss(a,b,c)\parallel d)$ or as $\tf(\tss(a,b,d)\parallel c)$.  In
this case, the two forget points are relatively close, and one of our
axioms makes it possible to jump directly from one parsing to the
other.  A similar situation arises with the graph given on the right.

The cornerstone of our completeness proof is the fact that two
parsings of a non-atomic full prime graph can always be rewritten so
as to agree on their forget point. This is Lemma~\ref{lem:fp:sync},
and most of the paper is devoted to proving it.

To do so, we first delimit a class of vertices which we call
\emph{anchors}. Those can be thought of as a generalisation of
cutvertices~\cite[Section 1.4]{DBLP:books/daglib/0030488}. When a full
prime graph has an anchor, we show that we can use it as a universal
agreement point (Lemma~\ref{lem:anchor:parsing}).

Unfortunately, there are graphs without anchors.  We call them
\emph{hard}, they can be thought of as specific unseparable
graphs~\cite[Section 1.4]{DBLP:books/daglib/0030488}. The rest of the
proof consists in a structural analysis of hard graphs of treewidth at
most three.
We show that they admit \emph{separation pairs}: every hard graph has
the shape on the left of Figure~\ref{fig:sep}, and every parsing
can be rewritten as a double forget on some of these separation
pairs (Lemma~\ref{lem:reaching:sep}). %
\begin{figure}[t]
  \centering
  \includegraphics{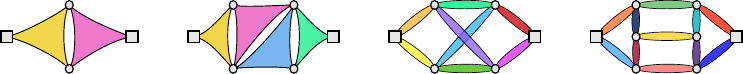}
  \caption{Separation pairs and graphs with several separation pairs.}
  \label{fig:sep}
\end{figure}
However, as illustrated on the right, separation pairs are not unique.
In order to finish the proof, we analyse how separation pairs
relate to each other.  There are easy cases like in the second graph
of Figure~\ref{fig:sep} where the diagonal separation pair connects
the two vertical ones, and previously encountered axioms make it
possible to conclude. By analysing the shape of graphs excluding the
triangle between their sources as a minor
(Proposition~\ref{prop:triangle_property}---note that cycles may still
appear in such graphs), we show that all other situations reduce to
one of the two graphs on the right of Figure~\ref{fig:sep}.  There,
the two outer-vertical pairs of inner vertices are the only separation
pairs, and they are disjoint. They correspond to
distinct parsings of the same graph, and we must include the corresponding
equations to complete our list of axioms.

\medskip
\iflong%
  For the sake of readability, some proofs and details are deferred to the appendix.
  Statement headers contain hyperlinks to ease navigation.
\else%
  For the sake of readability, some proofs and details are provided in the appendix of the full version~\cite{dhp:icalp24:tw3:hal}.
\fi%

\section{Graphs, treewidth}
\label{sec:graphs}

A \emph{ranked set} (or \emph{signature}) is a set where every element
has an associated natural number called its \emph{arity}. Throughout
the paper we fix a ranked set $\Sigma$ of \emph{letters}, which we
call the \emph{alphabet}.
Given a set $V$, we denote by $\LL(V)$ the ranked set of
duplicate-free lists over $V$, where the arity of a list is its
length.

We consider labelled and ordered hypergraphs with interfaces, defined as follows:
\begin{definition}
  \label{def:graphs}
  A \emph{graph} is a tuple $\tuple{V,E,n,l,i}$ where
  $V$ is a finite set of \emph{vertices},
  $E$ is a finite ranked set of \emph{edges},
  $i\in\LL(V)$ is the \emph{interface},
  and $n\colon E\to\LL(V)$ and 
  $l\colon E\to\Sigma$ are arity-preserving functions respectively
  giving the \emph{neighbours} and the \emph{label} of each edge.
\end{definition}
The \emph{elements} of a graph are its vertices and edges. The
vertices appearing in the interface of a graph are its
\emph{sources}. %
Vertices (resp. elements) which are not sources are called \emph{inner
  vertices (resp. elements)}.  %
A vertex is \emph{isolated} if it is not in the neighbourhood of any
edge. %
The \emph{size} of a graph $G$, written $|G|$, is its number of
elements. %
The \emph{arity} of a graph is that of its interface.
Two graphs $G,H$ are \emph{isomorphic}, written $G\giso H$, if their
vertex and edge sets are related by structure-preserving bijections
(cf. \citeapp{app:terms}{A}).
A graph is:
\begin{itemize}
\item \emph{empty} if its only elements are its sources;
\item \emph{atomic} if its only inner element is an edge, whose
  neighbours comprise all the sources.
\end{itemize}
\begin{toappendix}
  \label{app:terms}
  We first give a formal definition of graph isomorphism.
  \begin{definition}
    \label{def:giso}
    Two graphs $G=\tuple{V_G,E_G,n_G,l_G,i_G}$ and $H=\tuple{V_H,E_H,n_H,l_H,i_H}$ are \emph{isomorphic} if there are bijections $\varphi\colon V_G\to V_H$ and $\psi\colon E_G\to E_H$ such that
    \begin{itemize}
    \item $n_H\circ \psi = \LL\varphi \circ n_G$ (preservation of neighbourhoods),
    \item $l_H\circ \psi = l_G$ (preservation of labels),
    \item and $i_H = \LL\varphi (i_G)$ (preservation of the interface),
    \end{itemize}
    where $\LL\varphi$ denotes the pointwise extension of $\varphi$ to duplicate-free lists.
  \end{definition}
\end{toappendix}

Our choice of considering ordered hypergraphs with interfaces comes
from the need to be able to substitute graphs for edges.
A \emph{substitution} is an arity-preserving function from the
alphabet to graphs. %
Given a graph $G$ and a substitution $\sigma$, we write $G\sigma$ for
the graph obtained from $G$ by replacing all its $a$-labelled edges by
copies of the graph $\sigma(a)$, identifying pairwise the neighbours
of the edges with the interfaces of the copies. %
We say that $G$ \emph{has shape} $H$ when $G\giso H\sigma$ for some
substitution $\sigma$.  %
For instance, the first graph of Figure~\ref{fig:graphs} has the
following shapes (amongst other ones):
\begin{align*}
  \includegraphics{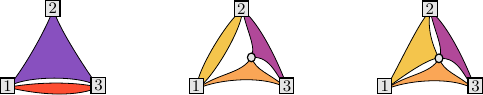}
\end{align*}
(Note that the edge orientation is irrelevant in a graph used as a
shape, as well as the labelling function---as long as it is injective,
which will always be the case in the present paper.)

\begin{proposition}
  \label{prop:equational}
  For all graphs $G,H$ and all substitutions $\sigma,\rho$, if
  $G\giso H$ and $\sigma(a)\giso\rho(a)$ for all letters $a$, then
  $G\sigma\giso H\rho$.
\end{proposition}
A graph in the sense of Definition~\ref{def:graphs} can be seen as a
simple graph, its \emph{skeleton}, by turning all its edges and its
interface into cliques (removing duplicate edges if necessary). %
This operation makes it possible to use the standard notion of
treewidth~\cite[Section 12.4]{DBLP:books/daglib/0030488}: the
\emph{treewidth} of a graph is that of its skeleton. Equivalently, we
can define it as follows. %
\begin{definition}\label{def:treewidth}
  A \emph{tree decomposition} of a graph is a tree whose nodes are
  labelled by sets of vertices, called \emph{bags}, such that
  \begin{itemize}
  \item there is a bag containing all sources,
  \item for each edge there is a bag containing its neighbours,
  \item for each vertex, the bags containing it form a subtree.
  \end{itemize}
  The \emph{width} of a tree decomposition is the size of its largest
  bag minus one. %
  The \emph{treewidth} of a graph is the minimal width of its possible
  tree decompositions.
\end{definition}
Observe that since the sources of a graph must be contained in a bag
for all tree decompositions, the treewidth of a graph is at least its
arity minus one. %
Therefore, at treewidth at most three, we only have graphs of arity
up to four. The same constraint holds for the edge arities.
\begin{propositionrep}
  \label{prop:tw:subst}
  If a graph $G$ and the graphs in the image of a substitution
  $\sigma$ all have treewidth at most $k$, then so does $G\sigma$.
\end{propositionrep}
\begin{proof}
  Starting from a tree decomposition $T$ for $G$, we obtain a tree
  decomposition for $G\sigma$ by adding, for each $a$-labelled edge
  $e$ of $G$, a copy $T_e$ of a tree decomposition for $\sigma(a)$,
  and an edge connecting a bag of $T$ containing the neighbours of $e$ to
  a bag of $T_e$ containing the sources of $\sigma(a)$. This
  construction does not create bags, and thus preserves their maximal
  size.
\end{proof}

\section{Graph operations}
\label{sec:ops}

We can now define the operations discussed in the introduction.
\begin{definition}
  \label{def:graph-op}
  Let $G,H$ be graphs of arity $k$.
  \begin{itemize}
  \item The \emph{parallel composition} of $G$ and $H$,
    $G\parallel H$, is the graph of arity $k$ obtained from the
    disjoint union of $G$ and $H$ by pairwise merging their sources.
  \item If $k\geq 1$, the \emph{forget} of $G$, $\tf G$, is the graph
    of arity $k-1$ obtained from $G$ by removing the last vertex from
    its interface (keeping it in the vertex set).
  \item The \emph{lift} of $G$, $\tl G$, is the graph of arity $k+1$
    obtained from $G$ by adding a new isolated vertex, and appending
    it to its interface.
  \item For a permutation $p$ of $[1,k]$, we denote by $pG$ the graph
    obtained from $G$ by permuting its interface according to $p$.
  \item We write $\top_k$ for the empty graph of arity $k$, omitting
    $k$ when it is clear from the context.
  \item For a letter $a\in\Sigma$, we also write $a$ for the atomic
    graph whose edge is labelled $a$ and has its neighbours appearing
    in the same order as in the interface.
  \end{itemize}
\end{definition}
These operations all arise as substitutions of well chosen small
graphs. By Proposition~\ref{prop:tw:subst}, this entails that
they preserve treewidth (except of course for lifts, which may
increase it when reaching the maximal arity).
It remains to establish a few structural properties, in order to show
that every graph of treewidth at most $k$ can be constructed from the
above operations up to arity $k+1$ (Proposition~\ref{prop:tw:terms} below).

A \emph{path} in a graph is a sequence of elements such that any two
consecutive elements consist of an edge and one of its neighbours. %
An \emph{inner path} is a path made only of inner elements, except
possibly for the endpoints.
\begin{definition}
  A graph is \emph{prime} if it is not empty and all its inner elements
  are connected by inner paths; it is \emph{full} if none of its
  sources are isolated.
\end{definition}
Observe that a graph is full prime iff all its elements are connected
by inner paths. Amongst the graphs in Figure~\ref{fig:graphs}, the
first one is full but not prime, the three other ones are prime, and
only the last one is full prime.

\begin{lemmarep}
  \label{lem:glift:decomp}
  Every graph is isomorphic to a permutation of lifts of a full graph.
  This decomposition is unique up to permutation and isomorphism.
\end{lemmarep}
\begin{proof}
  Starting from $G$, build a full graph $G'$ by removing all isolated
  sources both from the graph and from the interface. The graph $G$ can be
  obtained from $G'$ using as many lifts as there were isolated
  sources, followed by a permutation to move those isolated sources at
  their original places in the interface.
\end{proof}
\begin{lemmarep}
  \label{lem:prime:decomp}
  Every graph is isomorphic to a parallel composition of prime graphs.
  This decomposition is unique up to reindexing and isomorphism.
\end{lemmarep}
\begin{proof}
  The relation $C$ on inner elements such that $xCy$ when $x$ and $y$
  are connected via an inner path is an equivalence relation. Its
  equivalence classes seen as subgraphs provide the desired
  components.
\end{proof}
\noindent
We call \emph{(prime) components of a graph} the prime graphs occurring
in the latter decomposition. %
We call \emph{reduced components of a graph} the full prime graphs
obtained by removing isolated sources from its components.

\medskip

We now give two key properties of full prime graphs of bounded
treewidth. %
First, they are always atomic at maximal arity.
\begin{propositionrep}
  \label{prop:fp:atomic}
  Full prime graphs of treewidth at most $k$ and arity $k+1$ are
  atomic.
\end{propositionrep}
\begin{proof}
  Let $G$ be a full prime graph of treewidth at most $k$ and arity
  $k+1$. %
  We use the fact that graphs of treewidth at most $k$ cannot have the
  clique with $k+2$ vertices as a minor. If $G$ has an inner vertex
  $x$, then $x$ must be connected by inner paths to all vertices; we
  can thus contract all the inner vertices onto $x$, and obtain a
  clique with the sources and $x$, a contradiction. Thus the only
  inner elements of $G$ are edges; there must be exactly one since $G$
  is prime, and it must connect all sources since $G$ is
  full. Therefore $G$ is atomic.
\end{proof}
In the case of treewidth at most three, this means that when we
decompose a graph of arity four into prime components, then except for
a number of four-edges between the sources, we only get non-full
components: drawing the graph as a tetrahedron between its sources,
the non-atomic components are glued through the faces (or edges, or
vertices, or nothing), and the interior of the tetrahedron remains
empty. This corresponds to the characterisations of treewidth at most
$k$ graphs as partial $k$-trees~\cite{arnborg1986characterization}
or as subgraphs of chordal graphs
whose cliques are of size at most $k+1$~\cite{DBLP:journals/jal/RobertsonS86}.

Second, non-atomic full prime graphs have \emph{forget points}, which
we define as follows.
\begin{definition}
  When $x$ is an inner vertex of a graph $G$, we write $(G,x)$ for the
  graph obtained from $G$ by appending $x$ to its interface; if this
  graph has treewidth at most $k$ then we say that $x$ is a
  \emph{$k$-forget point of $G$}.
\end{definition}

\begin{toappendix}
  \begin{lemma}
    \label{lem:forget:points}
    Graphs of treewidth and arity at most $k$, and with an inner
    vertex, have $k$-forget points.
  \end{lemma}
  \begin{proof}
    Take a tree decomposition of width at most $k$. %
    If there is an isolated inner vertex, then it is a forget point. %
    Otherwise, the given inner vertex is in the neighbourhood of an
    edge, and there are bags containing inner vertices. %
    Take one of the shortest paths from a bag containing the sources
    to a bag containing some inner vertex $x$. By minimality, all bags
    along this path but the last one must contain only sources. %
    Add $x$ to all of them to obtain a tree decomposition for $(G,x)$
    of width at most $k$, establishing that $x$ is a $k$-forget point
    of $G$.
  \end{proof}
\end{toappendix}

\begin{propositionrep}
  \label{prop:fp:analysis}
  Non-atomic full prime graphs of treewidth at most $k$ have
  $k$-forget points.
\end{propositionrep}
\begin{proof}
  Direct consequence of Proposition~\ref{prop:fp:atomic} and
  Lemma~\ref{lem:forget:points}.
\end{proof}

\section{Terms and axioms}
\label{sec:terms}

We finally provide a notion of \emph{term} for denoting graphs. %
We use a multisorted syntax: the family $(\tset_k)_{k\in\NN}$ of sets
of terms of a given arity is defined inductively by the following
rules (where $p$ ranges over permutations of $[1,k]$ and $a$ over
letters of arity $k$):
\begin{align*}
  \frac{t,u\in \tset_k}{t\parallel u\in\tset_k}&&
  \frac{t\in \tset_k}{\tl t\in\tset_{k+1}}&&
  \frac{t\in \tset_{k+1}}{\tf t\in\tset_k}&&
  \frac{t\in \tset_k}{p t\in\tset_k}&&
  \frac{}{\top\in\tset_k}&&
  \frac{}{a\in\tset_k}
\end{align*}
This syntax matches the operations in Definition~\ref{def:graph-op}. %
Accordingly, we obtain a recursive arity-preserving function $\goft$
from terms to graphs.
We say that a term $t$ \emph{denotes} the graph $G$, or that $t$ is a
\emph{parsing} of $G$, when $\goft(t)\giso G$. %
For instance, the expressions below the graphs in
Figure~\ref{fig:graphs}, seen as terms, are parsings of these
graphs. %
Note that the number of inner vertices of a graph is the number of
forgets appearing in any of its parsings.

We shall sometimes mention terms and refer implicitly to their graphs;
for instance writing that a term is prime to mean that its graph is
so.

The \emph{width} of a term is the maximal arity of its subterms, minus
one.
\begin{proposition}
  \label{prop:tw:terms}
  A graph has treewidth at most $k$ iff it has a parsing of width at
  most $k$.
\end{proposition}
\begin{proof}
  The backward implication is proven by induction on terms, using
  Proposition~\ref{prop:tw:subst} and the observation that all
  operations arise as substitutions. %
  For the direct implication, we proceed by lexicographic induction on
  the size of the graph followed by $k+1$ minus its arity (recall that
  the arity is always lower or equal to the treewidth plus one).

  If the graph has isolated sources, we use
  Lemma~\ref{lem:glift:decomp} to write it as a permutation of lifts
  and proceed recursively. %
  Otherwise it is full, and we decompose it into primes via
  Lemma~\ref{lem:prime:decomp}:
  \begin{itemize}
  \item if there are no components then the graph is empty and has a
    trivial parsing;
  \item if there are at least two components then they are smaller, and we proceed
    recursively;
  \item otherwise the graph is (full) prime; either it is atomic and
    it admits a permutation of a letter as a parsing, or, by
    Proposition~\ref{prop:fp:analysis}, it can be written as the
    forget of a graph of the same size but increased arity, which we
    may parse recursively.\qedhere
  \end{itemize}
\end{proof}

The previous proposition holds for all bounds $k$ on the treewidth. %
Some of our results below also hold generically, and we will discuss
these generalisations in Section~\ref{sec:ccl}. Still, from this point
on we focus on treewidth at most three.
\begin{convention}
  In the remainder, we only work with graphs of treewidth at most
  three. We simply call \emph{terms} the terms of width at most three,
  and \emph{forget points} the $3$-forget points.
\end{convention}

A \emph{(term) substitution} is an arity-preserving function from the alphabet to terms.
Such a function $\sigma$ extends uniquely to a homomorphism $\hat\sigma$ from terms to terms:
  \begin{align*}
    \hat\sigma(t\parallel u) &\eqdef \hat\sigma(t)\parallel \hat\sigma(u)&%
    \hat\sigma(\tl t) &\eqdef \tl\hat\sigma(t)\\
    \hat\sigma(\top) &\eqdef \top&%
    \hat\sigma(\tf t) &\eqdef \tf\hat\sigma(t)\\
    \hat\sigma(a) &\eqdef \sigma(a)&%
    \hat\sigma(p t) &\eqdef p\hat\sigma(t)
  \end{align*}
A \emph{context of arity $i\to o$} is a term of arity $o$ with a single occurrence of a designated letter $h$ of arity $i$, called the \emph{hole}. Given such a context $c$ and a term $t$ of arity $i$, we write $c[t]$ for the term $c$ where the hole is replaced by $t$. Note that $c[t]=\hat\sigma_t(c)$ for the substitution $\sigma_t$ mapping $h$ to $t$ and fixing all other letters.

An \emph{equational theory} is an equivalence relation $R$ on terms, relating only terms of
the same arity, and which is closed under contexts and substitutions:
\begin{itemize}
\item $(t,u)\in R$ entails $(c[t],c[u])\in R$ for all contexts $c$ of appropriate arity, and
\item $(t,u)\in R$ entails $(\hat\sigma(t),\hat\sigma(u))\in R$ for all substitutions $\sigma$.
\end{itemize}

Given two terms $t,u$, we write $t\tiso u$ when $\goft(t)\giso\goft(u)$. %
Thanks to Proposition~\ref{prop:equational}, this relation $(\tiso)$ is an
equational theory, and our goal is to provide a finite list of axioms that
generates it.

We give such a list below; most axioms can be written explicitly, but
three of them relate rather large terms, which are best presented by
their graphs. This is why we rely on the following concept of forget axiom.

A \emph{forget axiom} for a graph $G$ with two forget points $x,y$ is
an equation of the shape $\tf t\equiv \tf u$ for some parsings $t$ of
$(G,x)$ and $u$ of $(G,y)$.

\newcommand\formatax[1]{\textcolor{lipicsGray}{\sffamily\bfseries\upshape\mathversion{bold}#1}}
\newcommand\tformatax[1]{\text{\formatax{#1}}}
\newcommand\ax[1]{\hyperlink{firstaxioms}{\formatax{A#1}}\xspace}%
\newcommand\axFS[1][]{\hyperlink{firstaxioms}{\formatax{FS$_{#1}$}}\xspace}%
\newcommand\axFK{\hyperlink{anchoraxioms}{\formatax{FK}}\xspace}%
\newcommand\axFX{\hyperlink{hardaxioms}{\formatax{FX}}\xspace}%
\newcommand\axFD{\hyperlink{hardaxioms}{\formatax{FD}}\xspace}%
\begin{definition}[Finite axiomatisation]
  \label{def:axioms}
  We write $\teq$ for the least equational theory containing the
  following axioms, for letters $a,b,c$ of appropriate arities%
  \hypertarget{firstaxioms}{:}
  \begin{enumerate}[{A}1.]
  \item $a\parallel (b\parallel c) \teq (a\parallel b)\parallel c$,
    $a\parallel b \teq b\parallel a$, and $a\parallel\top \teq a$;
  \item $pqa\teq (p\circ q)a$ for all permutations
    $p,q$, and $\id a \teq a$ where $\id$ is the identity permutation;
  \item $p(a\parallel b)\teq p a\parallel p b$ and
    $p\top\teq\top$ for all permutations $p$;
  \item $\tl(a\parallel b)\teq\tl a\parallel\tl b$
    and $\tl\top\teq\top$;
  \item $p\tf a \teq \tf\dot p a$ and
    $\tl p a \teq \dot p \tl a$ where $\dot p$ is the extension of a
    permutation $p$ of $[1,k]$ to $[1,k+1]$;
  \item $\tl\tf a\teq\tf r\tl a$ and
    $\tl\tl a\teq r\tl\tl a$ for the permutation $r$ that swaps the
    last two elements;
  \item $\tf a\parallel b \teq \tf(a \parallel\tl b)$;
  \item[\formatax{FS.}] $\tf\tf a\teq \tf\tf r a$ for the permutation
    $r$ that swaps the last two elements;
  \end{enumerate}
  as well as three forget axioms for the graphs with forget points \axFK,
  \axFX and \axFD in Figures~\ref{fig:anchor:axioms}\&\ref{fig:hard:axioms} (for some arbitrary choice of interface
  ordering, edge orientation, and injective edge labelling).
\end{definition}
\begin{figure}[t]
  \centering
  $\begin{array}{c@{\qquad\qquad}c@{\qquad\qquad}c@{\qquad\qquad}c}
    \includegraphics{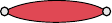} &
    \includegraphics{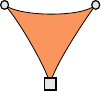} &
    \includegraphics{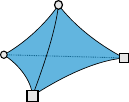} &
    \includegraphics{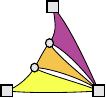} \\
    \tformatax{FS}_0 &
    \tformatax{FS}_1 &
    \tformatax{FS}_2 &
    \tformatax{FK} \\
  \end{array}$
  \caption{\hypertarget{anchoraxioms}{Forget axioms} for anchors (swap and kite); $x,y$ are the inner vertices.}
  \label{fig:anchor:axioms}
\end{figure}
\begin{figure}[t]
  \centering
  $\begin{array}{c@{\qquad\qquad}c}
    \includegraphics{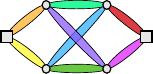} &
    \includegraphics{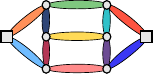} \\
    \tformatax{FX} &
    \tformatax{FD}
  \end{array}$
  \caption{\hypertarget{hardaxioms}{Forget axioms} for hard graphs (cross and domino); $x,y$ are the topmost inner vertices.}
  \label{fig:hard:axioms}
\end{figure}
The first eight items are universally quantified over all appropriate
arities. %
For instance, associativity of parallel composition (in \ax1) is an axiom
at each arity, and $a,b$ and $\top$ may have arity up to three in \ax4. %
Since we restrict globally to arities up to four, the list of axioms is
nevertheless finite.

Also note that even though the axioms are expressed using letters,
they yield laws which hold under all term substitutions since $\teq$
is defined as an equational theory.

\axFS comprises three axioms depending on the arity of $a$ ($2,3,$ or $4$), which we shall sometimes refer to as \axFS[i] (with $i=0,1,$ or $2$, respectively).
These are forget axioms for the first three graphs in
Figure~\ref{fig:anchor:axioms}, and we could have presented them as such.
We cannot express \axFS[3] as it would require an edge of arity five; \axFK intuitively is a weakened version of it, expressible at treewidth three.

We shall see at the end of Section~\ref{sec:completeness} that the
concrete choice of parsings for the forget axioms \axFK, \axFX, and \axFD is irrelevant thanks to \ax{1-7}. %
The interested reader may find concrete equations for them in \citeapp{app:axioms}{C}.
\begin{toappendix}
  \label{app:axioms}
  We give possible choices of concrete terms for the forget axioms in Definition~\ref{def:axioms}.
  To ease notation, we first define the following three derived operations:
  \begin{align*}
    \arraycolsep=20pt
    \begin{array}{ccc}
      u\bullet v\eqdef\tss(u,v,\emptyset)&
      \tc u\eqdef (12)u&
      \star(u,v,w)\eqdef\tss((21)\tl v,(13)\tl w,(23)\tl u) \\[1em]
      \mbox{\hspace{-6mm}\includegraphics{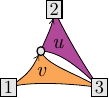}}&
      \mbox{\hspace{-5mm}\raisebox{6mm}{\includegraphics{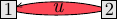}}}&
      \mbox{\hspace{-22mm}\includegraphics{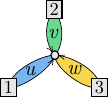}}\\
    \end{array}
  \end{align*}
  Observe that \axFK is just associativity of $\bullet$, i.e., $a\bullet (b\bullet c)\teq (a\bullet b)\bullet c$. We can write the two other axioms as follows, using the edge orientations and labelling given on the right:
  \begin{align*}
    \arraycolsep=4.5pt
    \begin{array}{lcc}
      \text{\axFX:}&
      \begin{gathered}
        \tff\big((23)\tll a \parallel
        (42)\tll b \parallel
        (14)\tl(\star(d,\tc e,c) \parallel \star(f,\tc h,g))\big)\\
        \teq\\
        \tff\big((13)\tll e\parallel(14)\tll h\parallel(24)\tl(
        \star(a,\tc g,\tc c)
        \parallel
        \star(b,\tc f,\tc d))\big)
      \end{gathered}&~
      \raisebox{-6mm}{\includegraphics{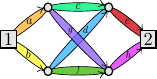}}\\[3em]
      \text{\axFD:}&
      \begin{gathered}
        \tff\big((23)\tll a\parallel(42)\tll b\parallel(14)\tl(
        \tss(
        \star(\tc i,\tc h,c),
        \star(f,\tc g,e),
        \star(d,\tc k,\tc j)
        ))\big)\\
        \teq\\
        \tff\big((13)\tll k\parallel(14)\tll h\parallel(24)\tl(
        \tss(
        \star(g,i,j),
        \star(b,\tc f,\tc d),
        \star(a,\tc c,\tc e)
        ))\big)
      \end{gathered}&~
      \raisebox{-6mm}{\includegraphics{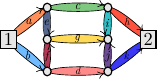}}\\
    \end{array}
  \end{align*}
\end{toappendix}

\begin{proposition}[Soundness]
  If $t\teq u$ then $t\tiso u$.
\end{proposition}
\begin{proof}
  Since $\tiso$ is an equational theory, it suffices to check that the
  two members of each axiom denote the same graph. This is by definition
  for the forget axioms, and a routine verification for the other
  ones.
\end{proof}

The axioms given textually in Definition~\ref{def:axioms} (\ax{1-7}
and \axFS) intuitively correspond to those given for arbitrary graphs
by Courcelle and Engelfriet~\cite[p.117]{DBLP:books/daglib/0030804},
restricted to terms of width at most three. (Modulo some translation
since our syntaxes and sets of operations differ.)  This finite
restriction is however incomplete for graphs of treewidth at most
three: in their completeness proof, Courcelle and Engelfriet need
axioms of the form \axFS[i] for arbitrary large arity $i$, even when
dealing with graphs of small treewidth. Somehow, we prove below that
the forget axioms \axFS[0], \axFS[1], \axFS[2], \axFK, \axFX, and
\axFD suffice at treewidth at most three to fulfil the role played
by this infinite family of axioms. Also observe that while \axFK
easily generalises to deal with an arbitrary treewidth $k$ (by
replacing the source touching the three edges by $k-2$ sources), this
is not the case for \axFX and \axFD.

\section{Outline of the completeness proof}
\label{sec:completeness}

Using the first seven axioms, we easily obtain the following
proposition. %
Together with Lemmas~\ref{lem:glift:decomp} and~\ref{lem:prime:decomp},
it makes it possible to normalise terms and to focus on full prime
ones.

\begin{propositionrep}
  \label{prop:tm:decomp}
  For all terms $t$, letters $a$, and graphs $G,H$, we have:
  \begin{enumerate}
  \item\label{tm:top} if $g(t)\giso \top$ then $t\teq\top$;
  \item\label{tm:at} if $g(t)\giso a$ then $t\teq a$;
  \item\label{tm:p} if $g(t)\giso pG$ then there is a parsing $u$ of
    $G$ such that $t\teq pu$;
  \item\label{tm:l} if $g(t)\giso \tl G$ then there is a parsing $u$
    of $G$ such that $t\teq \tl u$;
  \item\label{tm:par} if $g(t)\giso G\parallel H$ then there are
    parsings $u$ of $G$ and $v$ of $H$ such that $t\teq u\parallel v$.
  \end{enumerate}
\end{propositionrep}
\begin{appendixproof}
  All items but the third one are proven by induction on $t$.
  \begin{enumerate}
  \item Here $t$ cannot be a letter or a forget; in the other cases we
    use \ax{1-4}.
  \item We show that for all permutations $p$, $g(t)\giso pa$ implies
    $t\teq pa$; %
    $t$ cannot be $\top$, a forget or a lift; in the other cases we
    use \ax{1-2} and the previous item. %
    The initial statement follows by choosing the identity permutation.
  \item We take $u=p^{-1}t$ and use \ax2.
  \item We show by induction on $t$ that for all permutations $p$ and
    graphs $G$ such that $g(t)\giso p\tl G$, there is a parsing $u$ of
    $G$ such that $t\teq p\tl u$; the initial statement follows by
    choosing the identity permutation. The term $t$ cannot be a letter;
    if it is $\top$ then we take $u=\top$ and we conclude using
    \ax{3-4}; if it is a parallel composition we use the induction
    hypothesis and \ax{3-4} again. %
    It remains three cases:
    \begin{itemize}
    \item if $t=q t'$ for some permutation $q$, then
      $g(t')\giso q^{-1}p\tl G$. By induction we find a parsing $u$ of
      $G$ such that $t'\teq (q^{-1}\circ p)\tl u$. We deduce
      $t\teq p\tl u$ using \ax2;
    \item if $t=\tf t'$, then $g(t)=\tf(g(t'))\giso p\tl G$ so that,
      adding the source corresponding to the forget operation to $G$,
      there is a graph $G'$ such that $G\giso \tf G'$ and
      $g(t')\giso \dot p r\tl G'$, where $r$ is the permutation that
      swaps the last two elements, and $\dot p$ is the extension of
      $p$ (fixing the last element). %
      By induction we find a parsing $u$ of $G'$ such that
      $t'\teq (\dot p\circ r)\tl u$. %
      Then we have
      \begin{align*}
        t=\tf t'&\teq \tf (\dot p \circ r)\tl u\\
        \tag{by \ax2}
                &\teq \tf \dot p r\tl u\\
        \tag{by \ax5}
                &\teq p\tf r \tl u\\
        \tag{by \ax6}
                &\teq p\tl \tf u
      \end{align*}
    \item if $t=\tl t'$ and $g(t)=\tl(g(t'))\giso p\tl G$ then there
      are two cases:
      \begin{itemize}
      \item if the two lift operations in the latter graph isomorphism
        concern the same source in $G$, then it is the last one and $p$
        must fix the last element, i.e., $p=\dot q$ for some
        permutation $q$ and $g(t')\giso q G$. By the third item we
        find a parsing $u$ of $G$ such that $t'\teq q u$, and we have
        $t=\tl t'\teq \tl q u\teq \dot q \tl u = p\tl u$ using \ax5;
      \item otherwise, the source $s$ corresponding to the lift
        operation in $p\tl G$ must be isolated in $\goft(t')$, and the
        last element is not fixed by $p$. We assume it is mapped to
        the $i$th one. By Lemma~\ref{lem:glift:decomp} restricted to
        the source $s$ only there is a permutation $q$ such that
        $\goft(t')\giso q\tl G'$, the lift corresponds to $s$, and $q$
        maps the $i$th element to the last one. The permutation
        $r^{-1}\circ\dot q^{-1}\circ p$ must fix its last element so
        we assume it is equal to $\dot q'$ for some $q'$. By
        definition we have that $p=\dot{q}\circ r\circ\dot{q'}$. By
        induction and \ax2 we find a parsing $u$ of $G'$ such that $t'\teq q \tl u$.
        Then we have
        \begin{align*}
          t=\tl t'&\teq \tl q \tl u\\
          \tag{by \ax5}
                  &\teq \dot q \tll u \\
          \tag{by \ax6}
                  &\teq \dot qr\tll u \\
          \tag{by \ax2}
                  &\teq\dot qr\tl(q'\circ q'^{-1})\tl u \\
          \tag{by \ax5}
                  &\teq \dot qr\dot q'\tl q'^{-1}\tl u \\
                  &= p\tl q'^{-1}\tl u
        \end{align*}
      \end{itemize}
    \end{itemize}
  \item The cases where $t$ is $\top$, a letter, a lift or a
    permutation are easy. %
    In the case where $t=t_1\parallel t_2$, we split $G$ and $H$
    accordingly: $G\giso G_1\parallel G_2$, and
    $H\giso H_1\parallel H_2$ such that $g(t_1)\giso G_1\parallel H_1$
    and $g(t_2)\giso G_2\parallel H_2$; we use the induction
    hypothesis on $t_1$ and $t_2$ and we conclude using associativity
    and commutativity.

    The forget case is the interesting one. In that case, $t=\tf t'$,
    then $g(t)=\tf g(t')\giso G\parallel H$ and we may assume without
    loss of generality that the forget point appears in $G$:
    $G\giso \tf G'$ for some graph $G'$ such that
    $g(t')\giso G'\parallel \tl H$. By induction hypothesis we obtain
    parsings $u$ of $G'$ and $v'$ of $\tl H$ such that
    $t'\teq u\parallel v'$. By the fourth point we get a parsing $v$
    of $H$ such that $v'\teq \tl v$. We conclude using \ax7: $\tf u$
    and $v$ are parsings of $G$ and $H$ such that
    $t=\tf t'\teq \tf(u\parallel \tl v)\teq \tf u\parallel v$.
    \qedhere
  \end{enumerate}
\end{appendixproof}
\begin{proofsketch}
  All items but the third are proven by induction on $t$. %
  The first three items only need \ax{1-4}; %
  the fourth one needs \ax{5-6}; the last one rests on the fourth one
  and \ax7.
\end{proofsketch}

An item is patently missing from the previous proposition, for the
forget operation. %
In fact, a statement similar to the third and fourth items does not
hold for the forget operation: when $t$ is a parsing of $\tf G$,
nothing guarantees that $G$ has a parsing at all: its treewidth may be
four; we need to restrict to those cases where the forgotten vertex is
a forget point.

Given a term $t$ denoting a graph $G$ with an inner vertex $x$, we say
that $t$ \emph{reaches} $x$ if there is a parsing $t'$ of $(G,x)$ such
that $t\equiv \tf t'$. We can easily prove the following property.
\begin{lemma}
  \label{lem:full_prime_is_forget}
  Every non-atomic full prime term reaches some forget point.
\end{lemma}
\begin{proof}
  By induction on $t$, using Proposition~\ref{prop:tm:decomp} until we
  find a forget operation.
\end{proof}
However, this property is not enough, because it gives no control on
the forget point which is reached; instead, we would like to obtain:
\begin{statement}[Reaching forget points]
  \label{st:fp:parsing}
  Full prime terms reach all their forget points.
\end{statement}
This statement holds, but we do not know how to prove it directly: we
only get it \emph{a posteriori}, from the completeness. We prove a
variant of it in the next section (Lemma~\ref{lem:anchor:parsing}),
for \emph{anchors}.

Instead, the cornerstone of our proof is the following lemma: any two
parsings of a non-atomic full prime graph may reach a common forget
point.
\begin{lemma}[Forget point agreement]
  \label{lem:fp:sync}
  For all parsings $t,u$ of a non-atomic full prime graph, there is a
  forget point reached by both $t$ and $u$.
\end{lemma}
The rest of the paper is devoted to proving this lemma, with which we
conclude as follows.
\begin{theorem}[Completeness]
  \label{theo:completeness}
  For all parsings $t,u$ of a given graph, we have $t \teq u$.
\end{theorem}
\begin{proof}
  We follow the same pattern as in the proof of the forward
  implication of Proposition~\ref{prop:tw:terms}: we proceed by
  lexicographic induction on the size of the graph followed by $4$
  minus its arity.
  If the graph has isolated sources, we use
  Lemma~\ref{lem:glift:decomp} and
  Proposition~\ref{prop:tm:decomp}(\ref{tm:p},\ref{tm:l}) to rewrite
  the parsings into permuted lifts and proceed recursively. %
  Otherwise it is full, and we decompose it into primes via
  Lemma~\ref{lem:prime:decomp}:
  \begin{itemize}
  \item if there are no components then both parsings are equivalent to
    $\top$ by Proposition~\ref{prop:tm:decomp}\eqref{tm:top};
  \item if there are at least two components then we use
    Proposition~\ref{prop:tm:decomp}\eqref{tm:par} on both parsings
    and we proceed recursively.
  \item otherwise the graph must be (full) prime. Either it is atomic and
    we conclude by
    Proposition~\ref{prop:tm:decomp}(\ref{tm:at},\ref{tm:p}), or
    Lemma~\ref{lem:fp:sync} gives us two terms $t',u'$ such that
    $t \teq \tf(t')$, $t' \tiso u'$, and $\tf(u') \teq u$, and we can
    conclude since $t' \teq u'$ follows by induction hypothesis.
    \qedhere
  \end{itemize}
\end{proof}

We have only used Axioms \ax{1-7}
up to this point. We prove below
that those axioms are complete when there are few forget points.  Say
that a graph is \emph{easy} if each of its non-atomic reduced
components $G$ has only one forget point $x$, and in turn $(G,x)$ is
easy. (This definition is well-founded by the same lexicographic
ordering as we used in the above proof.)
\begin{proposition}
  \label{lem:easy}
  For all parsings $t,u$ of an easy graph, we have $t \teq u$.
\end{proposition}
\begin{proof}
  By adapting the previous completeness proof. %
  Lemma~\ref{lem:full_prime_is_forget} may replace
  Lemma~\ref{lem:fp:sync} when we attain non-atomic full prime graphs:
  those are reduced components of the starting graph and since those
  have only one forget point, all their parsings reach that one.
\end{proof}
None of the graphs in Figures~\ref{fig:anchor:axioms}
and~\ref{fig:hard:axioms} are easy: they are all full prime, each with
two forget points $x,y$ (plus two symmetrical ones for \axFX and
\axFD). %
Nevertheless, adding either $x$ or $y$ to their interface makes them
easy (cf. \iflong Lemma~\ref{lem:axioms:easy}\else\cite[Lemma~D.1]{dhp:icalp24:tw3:hal}\fi). %
This is why the precise choice of parsings does not matter when we use
them as forget axioms in Definition~\ref{def:axioms}.
\begin{toappendix}
  \begin{lemma}
    \label{lem:axioms:easy}
    For all graphs $G$ with forget points $x,y$ in
    Figures~\ref{fig:anchor:axioms} and~\ref{fig:hard:axioms}, both
    $(G,x)$ and $(G,y)$ are easy.
  \end{lemma}
  \begin{proof}
    By symmetry, it suffices to prove the statement for $(G,x)$.
    \begin{itemize}
    \item This is trivial for \axFS[i] and \axFK: for each of them, $(G,x)$ as only one inner vertex left.
    \item For \axFX, observe that $(\text\axFX,x)$ is full prime and its only forget point is the vertex depicted below $x$ (call it $x'$); then $(\text\axFX,x,x')$ has four reduced components: two edges and two graphs which are easy since they have only one inner vertex.
    \item For \axFD, $(\text\axFD,x)$ is full prime and its only forget point is the bottom-most vertex below $x$ (call it $x'$); then $(\text\axFD,x,x')$ has three reduced components: two edges with the source on the left, and the remainder of the graph. Again, there is a single forget point: the middle inner vertex on the right. From there we reach three reduced components, each with only one inner vertex.
    \end{itemize}
    (In those arguments, in order to see that an inner vertex cannot be a forget point, it often suffices to show that it would give a full prime graph of arity four which is not atomic.)
  \end{proof}
\end{toappendix}

\section{Series decompositions}
\label{sec:series}

As explained in the introduction, there is a \emph{series} operation
on graphs which plays an important role. We slightly generalise it
here, and we show how to use it to analyse graphs with a forget point.

Given $k$ graphs $G_1,\dots,G_k$ of arity $k$ and a graph $H$ of arity
$k+1$, we define the following graph of arity $k+1$, where $p_i$
denotes the permutation which swaps $i$ and $k+1$.
\begin{align*}
  \ts(G_1,\dots,G_k;H) \eqdef p_1\tl G_1 \parallel \dots \parallel p_k \tl G_k  \parallel H
\end{align*}
As is explicit from its definition, this operation on graphs is also a
derived operation on terms. %
We illustrate its behaviour at each arity in Figure~\ref{fig:series}.
We recover the operation from the introduction when the last argument
is empty, and using a forget operation: we have
$\tss(u,v,w)=\tf\ts(u,v,w;\top)$.
\begin{figure}[t]
  \centering
  \includegraphics{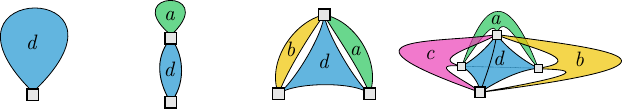}
  \caption{The series operations $\ts(;d)$, $\ts(a;d)$, $\ts(a,b;d)$, and $\ts(a,b,c;d)$.}
  \label{fig:series}
\end{figure}

We use this operation in order to decompose full prime graphs along a
given inner vertex.
\begin{proposition}[Series decomposition]
  \label{prop:sdeccomp}
  For all full prime graphs $G$ of arity $k$ with an inner vertex $x$,
  there are graphs $G_1,\dots,G_k,H$ such that
  $(G,x)\giso \ts(G_1,\dots,G_k;H)$ and
  \begin{enumerate}
  \item all components of $H$ are full, and
  \item if the $j$th source of a component of $G_i$ is isolated, then
    $i<j$.
  \end{enumerate}
  Such a decomposition is unique up to isomorphism.
\end{proposition}
We call $G_i$ the \emph{$i$th series argument}, and $H$ the
\emph{series factor} of such a decomposition (at $x$). %
\begin{proof}
  We decompose $(G,x)$ into prime components, which we classify
  according to their isolated sources.  Since $G$ is full prime, $x$
  is never isolated. %
  Full components go into the series factor. %
  Components where the first source is isolated go into the first
  series argument (under the permutation $p_1$, the first source gets
  swapped with $x$). %
  Components where the second source is isolated but not the first one
  go into the second series argument. \emph{Et caetera}.
\end{proof}

Also note that we can follow such decompositions at the term level,
modulo \ax{1-7}:
\begin{proposition}
  \label{prop:tm:sdeccomp}
  If a term $t$ reaches the last source of a graph of the form $\ts(G_1,\dots,G_k;H)$,
  then there are parsings $u_1,\dots,u_k,v$ of $G_1,\dots,G_k,H$ such
  that $t\teq\tf\ts(u_1,\dots,u_k;v)$.
\end{proposition}
\begin{proof}
  Consequence of Proposition~\ref{prop:tm:decomp}.
\end{proof}

\section{Anchors}
\label{sec:anchors}

In this section we define the concept of \emph{anchors}---inner
vertices with specific properties, and we prove the following variant
of Statement~\ref{st:fp:parsing}:
\begin{lemma}[Reaching anchors]
  \label{lem:anchor:parsing}
  Full prime terms reach all their anchors.
\end{lemma}
This lemma implies Lemma~\ref{lem:fp:sync} when the considered graph
has an anchor, by applying it to both parsings. Therefore, once
Lemma~\ref{lem:anchor:parsing} proved, it will only remain to prove
Lemma~\ref{lem:fp:sync} for anchor-free graphs
(Section~\ref{sec:hard}).

The definition of anchor is the following; we only consider them in
full prime graphs.
\begin{definition}
  \label{def:anchor_points}
  An \emph{anchor} in a (full prime) graph $G$ is an inner vertex $x$
  such that either:
  \begin{enumerate}
  \item there are no full prime components in $(G,x)$, or
  \item there is an edge whose neighbours comprise at least $x$ and all sources of $G$, or
  \item there are two or more full prime components in $(G,x)$.
  \end{enumerate}
\end{definition}
Using series decompositions, $x$ is an anchor whenever the series
factor at $x$ is empty, or has only one edge, or has at least two
components.

This last condition generalises the usual notion of
cutvertex~\cite[Section~1.4]{DBLP:books/daglib/0030488}: such anchors
at arity zero when there are only binary edges are exactly cutvertices
as usual. Removing them disconnects the graph. %
The second branch in our definition may only arise with sources and/or
hyperedges; it is convenient to have it for the present proof, but it
would probably be more natural to disallow it in other contexts. %
Thanks to this alternative, all inner vertices in the graphs in
Figure~\ref{fig:anchor:axioms} are anchors. In contrast, the graphs in
Figure~\ref{fig:hard:axioms} have no anchors.

We prove below that all forget points are anchors at arity three. %
The converse holds at all arities, but we only get it \emph{a
  posteriori}, as a consequence of Lemma~\ref{lem:anchor:parsing}.
\begin{proposition}
  \label{prop:arity_3_forget_points_are_central}
  At arity three, all forget points are anchors.
\end{proposition}
\begin{proof}
  If $x$ is a forget point of a graph $G$ of arity three, the full
  prime components of $(G,x)$ must be atomic by
  Proposition~\ref{prop:fp:atomic}. %
  Thus $x$ is an anchor, by either the first or the second condition
  in Definition~\ref{def:anchor_points}.
\end{proof}
We also find anchors easily at arities zero and one.
\begin{proposition}
  \label{prop:anchors013}
  Non-atomic full prime graphs of arity in $\set{0,1,3}$ have some
  anchor.
\end{proposition}
\begin{proof}
  Every inner vertex is an anchor at arity zero, either because it is isolated (anchor of the first kind), or because it belongs to some edge (anchor of the second kind). %
  Every immediate neighbour of the only source is an anchor at arity one (of the second kind). %
  At arity three we use
  Propositions~\ref{prop:arity_3_forget_points_are_central}
  and~\ref{prop:fp:analysis}.
\end{proof}

We first prove Lemma~\ref{lem:anchor:parsing} for graphs of arity
three, and then we show how to reduce the other cases to that one
(Section~\ref{ssec:anchors:lower}).

\subsection{Parsing on anchors at arity three}
\label{ssec:anchors:three}

A \emph{checkpoint} between two vertices $y,z$ of a graph $G$ is an
inner vertex $x$ such that every inner path from $y$ to $z$ goes
through $x$.  When $y,z$ are not mentioned explicitly, they are
implicitly assumed to be sources. (This definition coincides with that
from~\cite{DBLP:conf/mfcs/Cosme-LlopezP17} at arity two.)
Analysing the components of $(G,x)$ by the sources they touch, we see
that the checkpoints of a graph $G$ of arity three are the inner
vertices $x$ for which $G$ has the following shape:
\[
  \includegraphics{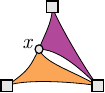}
\]
This characterisation shows that checkpoints are anchors of the first
kind in Definition~\ref{def:anchor_points}.
Another important property at arity three is that if there are two
anchors, then they must be checkpoints for the same sources:
\begin{toappendix}
  In the rest of the appendices we use some usual notations for
  paths. We call \emph{simple} a path without vertex repetition. Every
  (inner) path can be reduced to a simple (inner) path with the same
  endpoints.
  \begin{convention}
    In the rest of the appendices, we only work with simple paths,
    except if stated otherwise.
  \end{convention}
  Given a path $P$ and $x$ a vertex in $P$, we write $xP$ for the
  suffix of $P$ starting at $x$, and $Px$ prefix of $P$ stopping at
  $x$.  Given two inner disjoint paths $P$ and $Q$ such that $P$ ends
  at $Q$'s first vertex, we note $PQ$ for their concatenation. We
  write $\overline{P}$ for the reversal of a path $P$.  We extend
  these notations whenever they make sense on simple paths. For
  example we may write $xQy\overline{P}zR$. Given subgraphs $H$ and
  $K$ of a graph $G$, we call \emph{$H$-$K$ path} a path from an inner
  vertex of $H$ to one of $K$ that meets $H$ and $K$ only at its
  endpoints.  Most often $H$ will be just a vertex and $K$ will be a
  path.

  See \cite{DBLP:books/daglib/0030488} for a reference; there
  paths are called \emph{walks}, and simple paths just \emph{paths}.
  \begin{lemma}
    \label{lem:anchor:series_arg_checkpoint}
    Let $G$ be a full prime graph with two inner vertices $x,y$.
    If $x$ is in the $i$th series argument at $y$, then $y$ is a
    checkpoint between $x$ and the $i$th source in $G$.
  \end{lemma}
  \begin{proof}
    By definition the $i$th series argument of a series decomposition
    of $G$ does not contain its $i$th source. Let $P$ be an inner path
    from $x$ to the $i$th source of $G$. It must start in the $i$th
    series argument at $x$ and end outside at the $i$th source, as was
    just observed. As $P$ is inner, it can only get out from the $i$th
    series argument at $y$.
  \end{proof}
  \begin{lemma}
    \label{lem:anchor:two_checkpoints}
    Let $G$ be a graph with two inner vertices $x$ and $y$ and two designated
    sources $s_x$ and $s_y$.
    If $x$ is a checkpoint between $y$ and $s_y$ and $y$ a checkpoint
    between $x$ and $s_x$, then $x$ and $y$ are both checkpoints between
    $s_x$ and $s_y$.
  \end{lemma}
  \begin{proof}
    By symmetry, it suffices to show that $x$ is a checkpoint between $s_x$ and $s_y$.
    Let $P$ be an inner path from $s_x$ to $s_y$.
    Consider a shortest $y$-$P$ inner path $Q$ in $G$.
    Such a path exists as $G$ is full prime and $y$ is thus inner connected to any inner element of $P$.
    Say $Q$ is from $y$ to $z$ for some $z$ in $P$.
    The inner path $QzP$ goes from $y$ to $s_y$, so that it must contain $x$ by assumption.
    If $x$ belongs to $zP$ then we are done. Otherwise $x$ belongs to $Qz$.
    In that case, the inner path $xQz\overline{P}$ goes from $x$ to $s_x$ and must contain $y$ by assumption, which contradicts the minimality assumption about $Q$.
  \end{proof}
  \begin{lemma}
    \label{lem:sfa}
    In a full prime graph of arity three, the series factor at each anchor cannot contain inner vertices.
  \end{lemma}
  \begin{proof}
    We use the fact that graphs of treewidth at most three cannot have
    the clique $K_5$ over five vertices as a minor.
    Let $x$ be an anchor in a full prime graph $G$ of arity three and
    assume by contradiction that the series factor at $x$ has a
    component $C$ with some inner vertex $y$. There are two cases:
    \begin{itemize}
     \item $C$ is the only component of the series factor. Since $x$ is
      an anchor, $C$ must have an edge $e$ containing $x$ and the
      three sources. As $C$ is full prime and contains $y$, $e$ must contain
      another inner vertex of $C$ (possibly $y$): its arity is at least five, a
      contradiction.
    \item The series factor has (at least) one other component $D$.
      Since components of series factors are full prime, all inner
      vertices of $D$ are inner connected to all sources of $G$ and
      $x$, and similarly for the inner vertices of $C$.  Contract all
      inner vertices of $D$ onto $x$, and all inner vertices of $C$
      onto $y$.  This yields a minor with edges between $x$ and each
      and every source, similarly for $y$, and an edge between $x$ and
      $y$.  Since we can always assume a clique between the sources,
      we have obtained $K_5$ as a minor of $G$, a contradiction.
      \qedhere
    \end{itemize}
  \end{proof}
\end{toappendix}
\begin{propositionrep}
  \label{prop:central_central_central}
  At arity three, every full prime graph with two distinct anchors
  $x,y$ has the following shape:
  \[
    \includegraphics{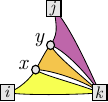}
  \]
\end{propositionrep}
\begin{proof}
  By Lemma~\ref{lem:sfa}, $x$ (resp. $y$) cannot be in the series
  factor of $G$ at $y$ (resp. $x$): it must be in a series argument of
  the series decomposition on $y$ (resp. $x$), say the $i$th
  (resp. $j$th) one.
  By Lemma~\ref{lem:anchor:series_arg_checkpoint} this makes
  $x$ (resp. $y$) a checkpoint between $y$ (resp. $x$) and the
  $i$th (resp. $j$th) source of $G$.
  We conclude with Lemma~\ref{lem:anchor:two_checkpoints}.
\end{proof}

\begin{proof}[Proof of Lemma~\ref{lem:anchor:parsing} at arity three]
  We proceed by induction on the size of the graph.
   
  Let $t$ be a parsing of a full prime graph $G$ of arity three, with
  an anchor $x$. %
  By Lemma~\ref{lem:full_prime_is_forget}, $t$ reaches some forget
  point $y$. %
  If $x=y$ then we are done. %
  Otherwise, $y$ is a second anchor by
  Proposition~\ref{prop:arity_3_forget_points_are_central}, and $G$
  has the shape given by
  Proposition~\ref{prop:central_central_central}.  Call $A$ the graph
  between $i$,$k$,$y$, and $B$ the one between $j,k,y$.  Since $t$
  reaches $y$ we have a parsing $t'$ of $(G,y)$ such that
  $t\teq\tf t'$. %
  By analysing the full prime decomposition of $(G,y)$, we can put
  $t'$ under the form $(j4)\tl u\parallel (i4)\tl v$, where $u$ is a
  parsing of $A$ and $v$ a parsing of $B$. %
  Then we observe that $x$ is a checkpoint in $A$, and thus an anchor,
  so that $u$ reaches $x$ by induction hypothesis. Therefore,
  $u\equiv \tf u'$ for some parsing $u'$ of $(A,x)$. %
  By analysing the full prime decomposition of $(A,x)$ as before, we
  deduce that $t$ reaches $x$ using \axFK.
\end{proof}

\subsection{Parsing on anchors at lower arities}
\label{ssec:anchors:lower}

We now finish the proof of Lemma~\ref{lem:anchor:parsing} by showing
how to reduce the other cases to that of arity three. We use for that
the two following lemmas. The first one intuitively makes it possible
to zoom on a specific inner vertex by going under a forget
operation. The second one states that under some conditions, an anchor
in $\tf^i G$ is also an anchor in $G$, allowing to use the case of
arity three.
\begin{lemmarep}
  \label{lem:forget-vers-x}
  Let $t$ be a term of arity up to two with an inner vertex $x$ in a
  full prime component. There is a term $u$ such that $t\teq\tf u$ and
  $x$ is either the last source of $u$ or an inner vertex of a full
  prime component of $u$.
\end{lemmarep}
\begin{proof}
  We do a proof by induction on the size of the component $C$ of $t$
  containing $x$. First we use Proposition~\ref{prop:tm:decomp} to
  find a parsing $t_x$ of $C$ and a term $t'$ such that
  $t\teq t_x\parallel t'$. By hypothesis $t_x$ is full prime. Hence we
  can use Lemma~\ref{lem:full_prime_is_forget} to find a term $u$ such
  that $t_x\teq\tf u$.

  If $u$ is a parsing of $(\goft t_x,x)$, or if $x$ is in a full prime
  component of $u$ then we conclude using \ax7.

  Otherwise $x$ is in a component $D$ of $u$ that is smaller in size
  than $C$ but not full prime. We use Proposition~\ref{prop:tm:decomp}
  to find terms $u_x$ and $u'$ and a permutation $p$ such that
  $u\teq p\tl^iu_x\parallel u'$ where $u_x$ is a parsing of the
  reduced component associated to $D$.

  As $u_x$ is full prime and contains $x$ we can use the induction
  hypothesis: there exists a term $v$ such that $u_x\teq\tf v$
  and either $v$'s last source is $x$ or $x$ is in a full prime
  component of $v$. Let $y$ be the last source of $v$; in subsequent
  term equations we write $\f_y$ for the forget operation
  that forgets $y$.

  At this point we have:
  \begin{align*}
    t&\teq\tf(p\tl^i\tf_y v\parallel u')\parallel t'\\
    \tag{by \ax6}
     &\teq\tf(p\tf_y(r\tl)^iv\parallel u')\parallel t'\\
    \tag{by \ax5}
     &\teq\tf(\tf_y\dot p(r\tl)^iv\parallel u')\parallel t'\\
    \tag{by \ax7}
     &\teq\tf\tf_y((\dot p(r\tl)^iv\parallel\tl u')\parallel\tl t')\\
    \tag{by \axFS}
     &\teq\tf_y~\tf r((\dot{p}(r\tl)^iv\parallel \tl
          u')\parallel\tl t')
  \end{align*}
  where $\dot p$ is the extension of $p$ fixing the last element
  and $r$ is the permutation swapping the last two elements (at various arities, here).

  If $x$ is $v$'s last source, then $x=y$ and the derivation above
  proves that there exists a term $w$ such that $t\teq\tf w$ and $w$'s last
  source is $x$. Otherwise $x$ is in a full prime component of $v$ and
  to end the proof we show that it is also in a full prime component
  of $w\eqdef\tf r((\dot{p}(r\tl)^iv\parallel \tl u')\parallel\tl t')$, i.e.
  that $x$ is inner connected to all sources in $w$.
  
  Call $z$ the source forgotten by the forget operation in $w$, and
  consider the shape of $(\goft t,y)$ given by the definition of $w$:
  \begin{center}
    \includegraphics{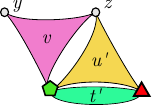}
  \end{center}
  In this picture, the red triangle is used to represent all sources of $t$ that are not
  sources of $v$ (those corresponding to the operation $(r\tl)^i$), and the green
  pentagon is used to represent the common sources of $t$ and $v$.

  As $x$ is in a full prime component of $v$, we easily get that $x$ is inner
  connected to $y$, $z$, and to all sources represented by the green pentagon.

  Moreover, $z$ is a checkpoint between $x$ and all sources represented by the red triangle.
  As $x$ is in a full prime component of $t$, there must be inner paths from $z$ to any
  such source in $t'$, and hence in $w$ too. By transitivity, $x$ is
  also inner connected to all sources represented by the red triangle.
\end{proof}

\begin{lemmarep}
  \label{lem:anchor_point_technical_lemma}
  For $i>0$ let $\tf^i G$ be a full prime graph of arity $k$ and $x,y$
  two distinct inner vertices in a full prime component $C$ of
  $G$. If $x$ is an anchor of $\tf^i G$ and is in one of the first
  $k$ series arguments of the series decomposition of $C$ at $y$, then
  $x$ is a checkpoint in $C$.
\end{lemmarep}
(Note that the only isolated sources of the components of the $i$ last
series arguments in the decomposition of $C$ as above, are the
$i$ forgotten sources in $\tf^iG$.)

Before we finish the proof of Lemma~\ref{lem:anchor:parsing}, also
observe that \axFS implies $\tf^jpu\teq \tf^jqu$ for all terms $u$
of arity $k$ and permutations $p,q$ agreeing on the first $k-j$
elements.

\begin{proof}
  We reason by case analysis on the conditions that may apply to make
  $x$ an anchor in $\tf^iG$. We note $x_1,\dots,x_i$ the sources of
  $G$ forgotten by the operations $\tf^i$ in the latter graph.

  Suppose that $x$ is in the $j$th series argument of $C$ at $y$. By
  Lemma~\ref{lem:anchor:series_arg_checkpoint}, $y$ is a checkpoint
  between $x$ and the $j$th source $s$ in $C$ and furthermore in $G$.

  In particular no edge can contain both $x$ and $s$. As $j\le k$, $s$
  is also a source of $\tf^i G$ and $x$ cannot be an anchor of the
  second kind in Definition~\ref{def:anchor_points}.

  If $x$ is an anchor of the first type in
  Definition~\ref{def:anchor_points} then the series factor of
  $\tf^i G$ at $x$ must be empty and inner vertices of $\tf^iG$ must
  be in the associated series arguments. This is the case of all
  $x_i$s. Assume, for example, that $x_1$ is in the $n$th series
  argument. By Lemma~\ref{lem:anchor:series_arg_checkpoint} $x$ is a
  checkpoint between $x_1$ and the $j$th source in $\tf^iG$. This
  obviously remains true in $C$.

  We can now assume that $x$ is an anchor of the third kind in
  $\tf^iG$: there are at least two full prime components in the series
  factor of $\tf^iG$ at $x$. At least one of these components does not
  contain $y$. Being full prime, this component provides us with an
  inner path $P$ from $x$ to the $j$th source $s$ in $\tf^iG$ that
  does not contain $y$. It cannot be an inner path in $G$ as in $G$,
  $y$ is a checkpoint between $x$ and $s$. As $G$ and $\tf^iG$ differ
  only on their sources, this means that $P$ must contain some $x_k$.

  Vertices $x_k$ and $y$ being in different components in the series
  decomposition of $\tf^iG$ at $x$, $x$ is a checkpoint between $x_k$
  to $y$ in $\tf^iG$, and it remains so in $G$ and hence in $C$ as all
  inner paths of $C$ are inner paths of both $G$ and $\tf^iG$.

  We have shown that in $C$, $x$ is a checkpoint between the source
  $x_k$ and $y$, and $y$ is a checkpoint between $x$ and the $j$th
  source. We conclude using Lemma~\ref{lem:anchor:two_checkpoints}.
\end{proof}

\begin{proof}[Proof of Lemma~\ref{lem:anchor:parsing}, at all arities]
  We prove the following generalisation, by induction on the
  lexicographic product of $|C|$ and $3-k$:
  \begin{quote}
    For all terms $t$ of arity $k\leq 3$, for all $i\leq k$, if $\tf^it$ is full prime with an anchor
     $x$ in a full prime component $C$ of $t$, then $\tf^it$ reaches $x$.
  \end{quote}
  If $k<3$ then we use Lemma~\ref{lem:forget-vers-x} to obtain $u$
  such that $\tf^it\teq\tf^{i+1}u$; if $x$ is the last source of $u$
  then we are done, otherwise we conclude by induction, since
  $x$ lies in a component of $C$.

  We now assume $k=3$, and we rewrite $t$ as
  $t_C\parallel t'$, where $t_C$ is the full prime component
  containing $x$. By Lemma~\ref{lem:full_prime_is_forget}, $t_C$
  reaches some forget point $y$.  Take a series decomposition
  $t_C\teq\tf\ts(t_1,t_2,t_3;u)$ of $t_C$ at $y$. The series factor
  $u$ cannot contain $x$ since all its components are atomic by
  Proposition~\ref{prop:fp:atomic}. Thus $x$ must be in a series
  argument.

  If $x$ is in one of the first $3-i$ series arguments, then by
  Lemma~\ref{lem:anchor_point_technical_lemma} $x$ must be a
  checkpoint, and thus an anchor of $t_C$. Using
  Lemma~\ref{lem:anchor:parsing} at arity three
  (Section~\ref{ssec:anchors:three}), we obtain that $t_C$ reaches
  $x$. So does $t$ by \ax7, and finally $\tf^it$ by \axFS.

  \smallskip\noindent
  \begin{minipage}{.82\linewidth}
    \hspace{1.5em}Otherwise, $x$ is in one of the last $i$ series arguments, say the
    $j$th one. The $j$th source $s_j$ is the only source that is necessarily
    isolated in this argument. It is a source that is forgotten by the
    operations $f^i$ (see the adjacent figure for an illustration). Our
    goal is to ``swap'' this $j$th source for $y$ in order to decrease
    the size of $C$. Indeed the new component containing $x$ will then
    be contained in $t_3$, which is smaller in size than $C$.
  \end{minipage}
  \begin{minipage}{.17\linewidth}
    \begin{flushright}
      \includegraphics{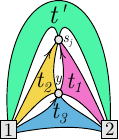}
    \end{flushright}
  \end{minipage}

  \medskip\noindent%
  Let $S$ be the set of non-isolated sources in the component $C'$ of
  $s(t_1,t_2,t_3;u)$ containing $x$.
  Using \axFS we get a term $t''$ such that $\tf^it\teq\tf^{|S|}t''$
  and $t''$ denotes $\goft(t)$ with the inner vertices in $S$ upgraded
  as sources.  By definition of $t''$, $C'$ is a full prime component
  of it. Furthermore, up to isolated sources, $C'$ is the same
  component as the one of $s(t_1,t_2,t_3;u)$ containing $x$. This
  gives $|C'|<|C|$ and we conclude by induction on $\tf^{|S|}t''$.
\end{proof}
At this point, we have used axioms \ax{1-7}, \axFS, and \axFK, but not \axFX and \axFD.

\section{Separation pairs}
\label{sec:hard}

We call \emph{hard} the non-atomic anchor-free full prime graphs.
Examples of such graphs were given in Figure~\ref{fig:sep}.
Hard graphs have arity two by Propositions~\ref{prop:anchors013} and~\ref{prop:fp:atomic}, and it only remains to prove Lemma~\ref{lem:fp:sync} for these graphs.

A \emph{forget pair} in a graph $G$ is a pair $(x,y)$ of inner
vertices such that $(G,x,y)$ has treewidth at most three. %
A parsing $t$ of a graph $G$ \emph{reaches a pair} $(x,y)$ of inner
vertices if there is some parsing $t'$ of $(G,x,y)$ such that
$t\teq \tff t'$. %
By definition, a term may only reach forget pairs. %
By \axFS, a term reaches $(x,y)$ iff it reaches $(y,x)$; it follows
that every term reaching $(x,y)$ reaches both $x$ and $y$.

Given two inner vertices $x,y$ of a full prime graph of arity two, we write
$y\diamond x$ when the graph has the shape on the left of
Figure~\ref{fig:order:sep}, and $y\glt x$ when the graph has the shape
on the right (equivalently, when $y$ is a checkpoint between $x$ and
some source). In the first case, we say that $(x,y)$ is a \emph{separation pair}.
\begin{figure}
  \centering
  \includegraphics{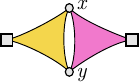} \hspace{3cm} \includegraphics{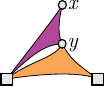}
  \caption{Separation pairs  $(y\diamond x)$, and partial order on vertices $(y\glt x)$.}
  \label{fig:order:sep}
\end{figure}
We prove a few properties about such inner vertices.
\begin{lemma}
  \label{lem:glt:reaches}
  If $y\glt x$ or $y \diamond x$ in a full prime graph $G$, then $y$ is an
  anchor in $(G,x)$, and every parsing of $G$ reaching $x$ also reaches $(x,y)$.
\end{lemma}
\begin{proof}
  In both cases, there is no full prime component in $(G,x,y)$, so
  that $y$ is indeed an anchor in $(G,x)$. Now if $t\teq \tf t'$ for
  some parsing $t'$ of $(G,x)$, then $t'$ reaches $y$ by
  Lemma~\ref{lem:anchor:parsing}: $t'\teq \tf t''$ for some parsing
  $t''$ of $(G,x,y)$. Thus $t$ reaches $(x,y)$.
\end{proof}
It follows that a full prime term reaches a separation pair iff it reaches any of its constituents.

We write $\sfa x$ for the series factor of a graph at some inner
vertex $x$. When the graph is hard, $x$ is not an anchor, so that $\sfa x$
must be full prime and cannot contain just one edge.
\begin{proposition}
  \label{prop:glt:wf}
  For all hard graphs, $\glt$ is well-founded.
\end{proposition}
\begin{proof}
  Observe that if $y\glt x$ then $|\sfa y|<|\sfa x|$.
\end{proof}

\begin{proposition}
  \label{prop:gltsep:sep}
  For all vertices $x,y,z$ such that $y\glt x$ and $x\diamond z$, we have $y\diamond z$.
\end{proposition}
\begin{proof}
  Consider an inner path between the two sources; we have to show that it visits either $y$ or $z$.
  It visits either $x$ or $z$ since $(x,z)$ is a separation pair. If it visits $x$ then it also visits $y$ since $y$ is a checkpoint between $x$ and one of the two sources.
\end{proof}
When the graph is hard in the previous proposition, it must have the following shape:
\vspace{-.2em}
\begin{center}
  \includegraphics{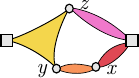}
\end{center}
\vspace{-.3em}
In such a case, we intuitively prefer the separation pair $(y,z)$ to $(x,z)$. This leads us to the notion of minimal separation pair.
A vertex $x$ is \emph{minimal} if there is no vertex $y$ such that $y\glt x$.
A pair of vertices is minimal when its two elements are so.

The proposition below makes it possible to get separation pairs out of minimal vertices.
\begin{proposition}
  \label{prop:form_of_hard_graphs}
  Let $x$ be a forget point in a hard graph.
  For all forget points $y$ of $\sfa x$, we have either $y\glt x$ or $y\diamond x$.
\end{proposition}
\begin{proof}
  Let us classify the components of $(G,x,y)$ by their
  isolated sources. Full components cannot exist
  as by Proposition~\ref{prop:fp:atomic} they would be edges and both
  $x$ and $y$ would be anchors in $G$. So $G$ must have the
  following shape:
  \[
    \includegraphics{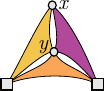}
  \]
  If there is no inner path avoiding $y$ between the sources in the bottom component, then $y\diamond x$.
  If there is no inner path avoiding $y$ from $x$ to one of the sources, then $y\glt x$.
  Otherwise there are at least two full prime components in $(G,y)$, making $y$ an anchor of $G$, which contradicts $G$ being hard.
\end{proof}
It follows easily that every hard graph has some separation pair---even a minimal one, which will be useful to reduce the number of cases to study.
However, we need to be more precise and to keep track of terms.
\begin{lemma}
  \label{lem:reaching:sep}
  Every hard term reaches some minimal separation pair.
\end{lemma}
\begin{proof}
  Let $t$ be a hard term.
  By Lemma~\ref{lem:full_prime_is_forget} $t$ reaches some forget point $x$, which we can choose minimal by Proposition~\ref{prop:glt:wf} and Lemma~\ref{lem:glt:reaches}.
  Accordingly, let $t'$ be a parsing of $(G,x)$ such that $t\teq \tf t'$.
  As $\sfa x$ is full prime and non-atomic, it has a forget point $y$.
  Since $x$ is minimal, $(x,y)$ is a separation pair by Proposition~\ref{prop:form_of_hard_graphs}.
  We can choose $y$ to be minimal by Proposition~\ref{prop:gltsep:sep}, and $t$ reaches $(x,y)$ by Lemma~\ref{lem:glt:reaches}.
\end{proof}

As before with forget points, separation pairs are not unique (even
minimal ones), and we need to show how to move from one separation
pair to another.

When studying the possible shapes of a hard graph $H$ with distinct
separation forget pairs, we end up with a few shapes $S$ such as the
following one:
\begin{align*}
  \includegraphics{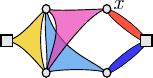}
\end{align*}
In such situations, we know that $(H,x)$ has treewidth at most three,
but $(S,x)$ has treewidth four. Therefore, for an appropriate notion
of \emph{minor}, $(S,x)$ cannot be a minor of $(H,x)$: the treewidth
may not increase when taking minors. We use this information in order
to refine the actual shape of $H$.

We use \emph{sourced simple graphs} in order to define such a notion
of minor.  Those are triples $(V,E,S)$ made of a set $V$ of
\emph{vertices}, a set $E$ of unordered binary \emph{edges}, and a
subset $S\subseteq V$ of \emph{sources}.  We write $K_k$ for the
clique over $k$ sources.  \emph{Tree decompositions} and
\emph{treewidth}~\cite[Section 12.4]{DBLP:books/daglib/0030488} are
adapted to sourced simple graphs by requiring the subset of sources to
be contained in some bag. A \emph{sourced minor of a (sourced simple) graph} is a
(sourced simple) graph obtained from it by a sequence of the following
operations: remove an edge, contract an edge, remove an isolated
vertex.  Those operations do not increase the treewidth.

The \emph{footprint} of a graph is the sourced simple graph obtained from it by replacing each edge by a clique over its neighbours, forgetting labels, and turning the interface into a mere set---note that we do not add a clique on the sources.
A graph and its footprint have the same tree decompositions, and thus the same treewidth.
A \emph{sourced minor of a graph} is a sourced minor of its footprint.

Note that graphs excluding $K_3$ as a sourced minor need not be acyclic, as cycles may occur away from the sources.
The point of using sourced minors is that they behave well with respect to substitutions, and thus shapes:
\begin{proposition}
  \label{prop:subst_minor}
  Given a graph $G$ and a substitution $\sigma$, if $K_k$ is a sourced minor
  of $\sigma(a)$ for all letters $a$ of arity $k$, then the footprint of
  $G$ is a sourced minor of $G\sigma$.
\end{proposition}
Now observe that full prime graphs of arity two admit the sourced edge $K_2$ as a
sourced minor.  At arity three, we have the following property instead.
\begin{toappendix}
  The following three lemmas (Lemmas~\ref{lem:three_paths_triangle}, \ref{lem:sourced_star}, \ref{lem:sep_pair_checkpair}), as well as Proposition~\ref{prop:triangle_property}, hold without any restriction on the treewidth.
  \label{app:hard}
  \begin{lemma}
    \label{lem:three_paths_triangle}
    Let $G$ be a graph of arity three with three inner paths $P_1$, $P_2$,
    and $P_3$, respectively from the first to the second, the second to
    the third, and the third to the first sources.
    If no vertex is shared by the three paths then $G$ admits $K_3$ as a sourced minor.
  \end{lemma}
  \begin{proof}
    Endpoints and vertices of paths are preserved by taking the
    footprint of a graph: the lemma holds for a graph whenever it
    holds for its footprint. Hence we may assume that $G$ is a sourced
    simple graph.

    We write $s_1,s_2,s_3$ for the three sources of $G$.

    We reason by induction on the sum of the lengths of the paths, with the
    \emph{length} of a path being the number of inner elements it
    contains.

    If $P_1$, $P_2$, and $P_3$ are all of length $1$ then they
    are reduced to edges connecting pairs of sources. Said otherwise,
    $G$ contains $K_3$ as a subgraph, and by extension as a minor.

    Otherwise, at least one of the paths $P_i$, say $P_1$, has an inner vertex.
    We note $e$ for the first edge and $x$ for the
    first inner vertex of $P_1$.

    If $x$ does not appear in either $P_2$ or $P_3$ then we contract
    $e$ in $G$, obtaining a graph $G/e$.  This shortens $P_1$, while
    retaining the hypothesis: we conclude by induction on $G/e$.

    Otherwise, $x$ appears in, say, $P_3$. We distinguish two possibilities:
    \begin{enumerate}
    \item $x$ is not the last inner vertex of $P_3$. In this case, we
      replace $P_3$ by the shorter path $P_3xes_1$. Path $P_1$, $P_2$,
      and $P_3xes_1$, still have no vertex common to all three: so we conclude
      by induction hypothesis.
    \item $x$ is the last inner vertex of $P_3$, and is followed by an
      edge $e'$: $P_3=P_3xe'$. Note that $e=e'$ is possible but not
      necessary. We then contract both $e$ and $e'$ and conclude again
      by induction hypothesis on the resulting minor. (Note that
      $x$ not being in $P_2$ guarantees that contracting $e$ and $e'$
      does not modify $P_2$).\qedhere
    \end{enumerate}
  \end{proof}

  \begin{lemma}
    \label{lem:sourced_star}
    Let $G$ be a full prime graph of arity $k$. For all $i\le k$, $G$
    admits the following sourced simple graph $H_i$ as a sourced
    minor: $H_i$ has the same sources as $G$, and its only inner
    elements are $k-1$ edges between the $i$th source and the other
    sources.
  \end{lemma}
  \begin{proof}
    If $G$ has no inner vertex then it must be atomic, it admits the clique on the $k$ sources as a sourced minor, and thus in particular the graphs $H_i$ from the statement.
    Otherwise, starting from the footprint of $G$, contract all inner vertices onto a single one, say $x$, which is connected to all sources since $G$ was full prime. For each $i$, it suffices to contract the edge between $x$ and the $i$th source to obtain $H_i$ as a sourced minor.
  \end{proof}
\end{toappendix}
\begin{propositionrep}
  \label{prop:triangle_property}
  Every graph of arity three
  either has the sourced triangle $K_3$ as a sourced minor, or has one of the two following shapes:
  \begin{center}
    \includegraphics{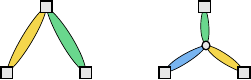}
  \end{center}
\end{propositionrep}
\begin{proof}
  Suppose $G$ is a graph of arity three that does not admit $K_3$ as a
  sourced minor.

  If $G$ has no full prime components, then it has the triangle as a
  shape. Not having $K_3$ as a sourced minor implies that some
  edge of this shape does not have any full prime component in $G$,
  making it possible to refine the triangle shape into the left-hand
  side one in the statement.

  If $G$ has two or more full prime components, we can use
  Lemma~\ref{lem:sourced_star} on these two components and two distinct
  sources of $G$ to get $K_3$ as sourced minor.

  It remains to deal with the case where $G$ is full prime.  We prove
  in that case that there is a vertex $x$ which is a checkpoint
  between all pairs of sources: it follows easily that $G$ has the
  right-hand side shape from the statement.
  
  Write $s_1$, $s_2$, and $s_3$ for the sources of $G$.  Let $P_1$,
  $P_2$, and $P_3$ be three paths in $G$ respectively from $s_1$ to
  $s_2$, $s_2$ to $s_3$, and $s_3$ to $s_1$, each of minimal length.
  By Lemma~\ref{lem:three_paths_triangle}, there is a vertex $x$
  shared by these three paths.

  Since segments such as $P_1x$ and $\overline{P_3}x$ share their
  endpoints, they must have equal length: otherwise one could be
  replaced by the other to obtain a shorter triple of paths.  Without
  loss of generality, we may actually chose them to be equal.
  Accordingly, we write $Q_1$ for $P_1x=\overline{P_3}x$, $Q_2$ for
  $P_2x=\overline{P_1}x$, and $Q_3$ for $P_3x=\overline{P_2}x$. Note
  that $Q_i$ is a path from $s_i$ to $x$.

  The paths $Q_i$ are pairwise inner disjoint. Indeed, would $Q_1$ and
  $Q_2$ share an inner vertex (for instance), then $P_1=Q_1\overline{Q_2}$
  would visit that vertex twice and would not be minimal.

  We finally prove that all inner paths between any two sources must
  contain $x$. Let $Q$ be an inner path between two sources, say $s_1$
  and $s_2$ without loss of generality, and assume by contradiction
  that $Q$ avoids $x$.

  By Lemma~\ref{lem:three_paths_triangle}, there is a vertex $y$
  shared by $Q$, $P_2$, and $P_3$. Since $Q_1$ and $Q_2$ are inner
  disjoint, $y$ must belong to $Q_3$.

  The path $Q$ starts at $s_1$, ends at $s_2$, and meets $Q_3$ at
  $y$. It cannot always remain on $Q_1$ as it does not contain
  $x$. Let $z$ be the last vertex of $Q$ in common with $Q_1$, and
  $z'$ be the first vertex after $z$ on $Q$ that meets $Q_j$ for
  $j\in\{2,3\}$. Note that the inner part of $R=zQz'$ is disjoint
  from all $Q_i$s. As $z$ and $z'$ cannot be $x$, we are in the
  following situation:
  \begin{center}
    \includegraphics{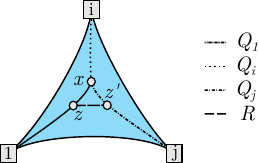}
  \end{center}
  It follows easily, using the fact that most drawn paths are disjoint,
  that $K_3$ is a sourced minor of $G$, a contradiction.
\end{proof}
Let us continue our example and come back to the previous hard graph $H$.
Assume that the binary edges in the given shape are instantiated by
full prime graphs, and consider the three graphs instantiating the
ternary edges.  If those three graphs have the triangle as a sourced minor,
then the footprint of $(S,x)$ must be a sourced minor of $(H,x)$, which is not
possible. Therefore, at least one of them must have one of the shapes
given by Proposition~\ref{prop:triangle_property}. By considering the
different possibilities and continuing the same kind of reasoning, we
end up proving that in this case, $H$ must have shape \axFX. In the
general case, we obtain the following classification.

\begin{toappendix}
  \begin{lemma}
    \label{lem:sep_pair_checkpair}
    Let $(x,y)$ be a separation pair in a graph $G$ of arity two. All
    paths from one source to the other in $G$ must visit either
    $x$ or $y$. If furthermore $G$ is hard then there exists a
    path that contains $x$ (resp. $y$) but not $y$ (resp. $x$).
  \end{lemma}
  \begin{proof}
    By definition of separation pairs, and using the fact that hard graphs are full prime and may not have checkpoints.
  \end{proof}
\end{toappendix}
\newcommand\scase[1]{\hyperlink{scases}{\formatax{Case #1}}\xspace}%
\newcommand\zcase[1]{\hyperlink{zcases}{\formatax{Case Z#1}}\xspace}%
\begin{propositionrep}
  \label{prop:two-sep-pairs}
  Let $G$ be a hard graph with two minimal separation forget pairs
  $(x,y)$ and $(x',y')$.
  Possibly up to swapping $x$ and $y$, and/or $x'$ and $y'$, either $(x,y')$ is a separation
  pair $(\tformatax{Case Z})$, or $G$ has one of the following shapes, where in the leftmost case, $z$
  is minimal and $x',y'$ appear in the rightmost component%
  \hypertarget{scases}{:}
  \begin{align*}
    \arraycolsep=20pt
    \begin{array}{ccc}
      \includegraphics{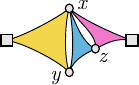}&
      \includegraphics{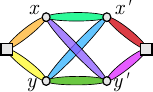}&
      \includegraphics{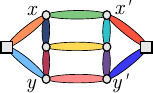}\\
      \tformatax{Case IH}&
      \tformatax{Case X}&
      \tformatax{Case D}\\
    \end{array}
  \end{align*}
\end{propositionrep}
The proof we provide in \citeapp{app:hard}{F} actually establishes that in the case where $(x,y')$ is a separation pair, then $G$ has one of the following shapes%
\hypertarget{zcases}{:}
\begin{align*}
  \arraycolsep=20pt
  \begin{array}{ccc}
    \includegraphics{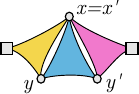}&
    \includegraphics{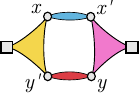}&
    \includegraphics{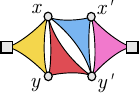}\\
    \tformatax{Case Z1}&
    \tformatax{Case Z2}&
    \tformatax{Case Z3}\\
  \end{array}
\end{align*}
\begin{proof}
  Let $G$ be a hard graph with $(x,y)$ and $(x',y')$ two
  minimal separation forget pairs. If $\set{x,y}$ meets $\set{x',y'}$ then we are trivially in \scase Z (in fact, \zcase1 by analysing paths between the two sources).
  Thus we assume in the sequel that $\set{x,y}$ and $\set{x',y'}$ are disjoint.

  We consider the graph from the point of view of the first separation pair $(x,y)$: by
  hypothesis $G$ has the following shape:
  \begin{center}
    \includegraphics{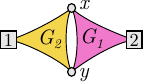}
  \end{center}
  We respectively write $s_1$ and $s_2$ for the first and second sources.
  We assume by symmetry that $x'$ appears in $G_1$, and we proceed by case analysis on whether $y'$ also appears together with $x'$ in $G_1$, or on the opposite side of $(x,y)$, in $G_2$.


  \subparagraph{1. when $x'$ is in $G_1$ and $y'$ is in $G_2$.}
  First observe that either $x'$ is a checkpoint between $x$
  and $s_2$ in $G_1$ or $y'$ is a checkpoint between $s_1$ and $x$ in
  $G_2$. Indeed, if we had both an $y'$-free path $P$ from $s_1$ to
  $x$ and an $x'$-free path $Q$ from $x$ to $s_2$, then their
  concatenation $PQ$ would be a path from $s_1$ to $s_2$ contradicting
  Lemma~\ref{lem:sep_pair_checkpair}.
  Therefore, either $(x',y)$ or $(x,y')$ is a separation pair, and we are in \scase Z.
  We show below that we are actually in \zcase2.

  Up to permutation of $x'$ and $y'$, assume that $x'$ is a checkpoint
  between $x$ and $s_2$ in $G_1$.
  By symmetry, we also have that either $x'$ is a checkpoint between $y$
  and $s_2$ in $G_1$ or $y'$ is a checkpoint between $s_1$ and $y$ in
  $G_2$. We now prove that only the latter case is possible. By
  contradiction suppose that $x'$ is a checkpoint both between $s_2$ and
  $y$ and between $s_2$ and $x$. Consider $P$ a $s_1$-$s_2$ path. By
  Lemma~\ref{lem:sep_pair_checkpair}, $P$ contains either $x$ or $y$.
  In the respective cases, $xP$ or $yP$ must contain $x'$ as $x'$ is a
  checkpoint between $x$ and $s_2$ and between $y$ and $s_2$ in $G_1$.
  Hence $x$ is a checkpoint in $G$, contradicting $G$ being hard.

  We end up with the following shape for $G$:
  \begin{center}
    \includegraphics{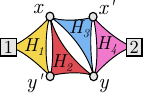}
  \end{center}

  Assume an inner $x$-$y$ path $P$ in either $H_2$ or $H_3$. Note that
  $(x,y')$ and $(x',y)$ are separation pairs. By
  Lemma~\ref{lem:sep_pair_checkpair} let be an inner path $Q$ (resp. $R$)
  from $s_1$ to $s_2$ containing $x$ (resp. $y$) but not $y'$ (resp.
  $x'$). The concatenation $QxPyR$ is a $s_1$-$s_2$ path in $G$
  containing neither $x'$ nor $y'$, contradicting
  Lemma~\ref{lem:sep_pair_checkpair}.

  Hence inner $x$-$y$ paths in $H_2$ or $H_3$ do not exist, meaning
  there are no reduced components of either $H_2$ or $H_3$ with both
  $x$ and $y$ as sources: we fall into \zcase2.

  \subparagraph{2. when $x'$ and $y'$ are both in $G_1$.}
  We first prove that $s_2$ is isolated in all components of $G_1$ not
  containing $x'$ nor $y'$.  Let $C$ be such a component ans suppose
  by contradiction that $s_2$ is not isolated.  If $x$ and $y$ are
  both isolated in $C$, then $C$ is actually a component of $G$,
  contradicting the fact $G$ is prime.  Otherwise, assume, without
  loss of generality, that $x$ is not isolated in $C$. Let $Q$ be an
  $x$-$s_2$ path in $C$. By Lemma~\ref{lem:sep_pair_checkpair}, let
  $P$ be an $s_1$-$s_2$ path in $G$ containing $x$ but not $y$. The
  path $PxQ$ is from $s_1$ to $s_2$, and contains neither $x'$ nor
  $y'$ as it intersects $H$ only at $x$ and $s_2$, a contradiction by
  Lemma~\ref{lem:sep_pair_checkpair}.

  Without loss of generality, we can move all those components from
  $G_2$ to $G_1$, so that either $G_1$ becomes prime and contains both
  $x'$ and $y'$, or $G_1$ has exactly two prime components,
  respectively containing $x'$ and $y'$.

  \subparagraph{2.1. when $G_1\giso H_1\parallel H_2$ with $x'$ in
    $H_1$, $y'$ in $H_2$, and $H_1,H_2$ prime.} We prove that $x'$ is
  a checkpoint between $x$ and $s_2$ in $H_1$. Up to replacing $x$ by
  $y$, and/or $x'$ by $y'$, this gives that $x'$ is also a checkpoint
  between $y$ and $s_2$ in $H_1$, and the same properties for $y'$ but
  in $H_2$.

  Consider an $x$-$s_2$ inner path $P$ in $H_1$. By
  Lemma~\ref{lem:sep_pair_checkpair}, let be $Q$ an $s_1$-$s_2$ inner
  path in $G$ that contains $x$ but not $y$. The path $QxP$ is from
  $s_1$ to $s_2$, and cannot contain $y'$, as it intersects $H_2$ only
  at $x$ and $s_2$. By Lemma~\ref{lem:sep_pair_checkpair}, $(x',y')$
  being a separation pair, $QxP$ must contain $x'$. It can do so only
  in $P$ as $x'$ is inner in $H_1$.

  We end up with the following shape for $G$:
  \begin{center}
    \includegraphics{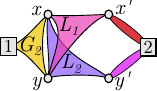}
  \end{center}
  Name $L$ this shape. A simple verification shows that $(L,x')$ has
  treewidth $4$. Note how, as $G$ is hard, the $2$-edges in $L$ must
  have full prime components in $G$. By
  Proposition~\ref{prop:subst_minor}, if all subgraphs of $G$
  corresponding to the $3$-edges in $L$ had $K_3$ as a sourced minor,
  $(G,x')$ would have the footprint of $(L,x')$ as a minor,
  contradicting the treewidth of the former graph. Hence one at least
  of $G_2$, $L_1$, or $L_2$ has a shape given by
  Proposition~\ref{prop:triangle_property}. Say this applies to
  $G_2$. The other cases are treated similarly and they all lead to the
  same final shape (\scase X).

  A simple verification shows that the only possibility given by
  Proposition~\ref{prop:triangle_property} for $G_2$ that keeps $G$ a
  hard graph and $x$ and $y$ minimal is the following:
  \begin{center}
    \includegraphics{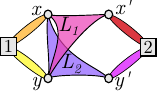}
  \end{center}
  Name $S$ this shape. Again $(S,x')$ has treewidth $4$, and both
  $L_1$ and $L_2$ cannot have $K_3$ as a sourced minor by
  Proposition~\ref{prop:subst_minor}. By symmetry, we can assume that
  a shape from Proposition~\ref{prop:triangle_property} applies to
  $L_1$. There is again only one possibility that keeps $G$ hard and
  $x'$ minimal:
  \begin{center}
    \includegraphics{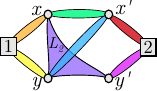}
  \end{center}
  Once again this shape with $x'$ as an added source has treewidth
  $4$, by Proposition~\ref{prop:subst_minor} $L_2$ cannot have $K_3$
  as a sourced minor, and the only possibility given by
  Proposition~\ref{prop:triangle_property} for $L_2$ gives the shape of \scase X for $G$.

  \subparagraph{2.2. when $G_1=H$ is prime.}
  As $x$ and $y$ are forget points, $H$ has treewidth at most $3$ as a
  component of $(G,x,y)$. Furthermore, $G$ being hard, $H=G_1$ must be
  full prime. As it contains at least two inner vertices $x'$ and
  $y'$, it must have a forget point $z$. Assume $z'\prec z$, i.e. $z'$
  is a checkpoint between $z$ and a source of $G$. We prove that $z'$
  is a forget point of $H$. First, we show that $z'$ is inner in $H$.
  If $z'$ is a checkpoint between $z$ and $s_2$, then this is obvious.
  Otherwise $z'$ is a checkpoint between $z$ and $s_1$. As $H$ is full
  prime, there are inner paths in $H$ from $z'$ to either $x$ or $y$.
  By Lemma~\ref{lem:sep_pair_checkpair}, there are inner paths in
  $G_2$ from $x$ (resp. $y$) to $s_1$ not containing $y$ (resp. $x$).
  The concatenation of the latter with the former proves that neither
  $x$ nor $y$ can be $z'$, and that if $z'$ is in $G_2$, then it would
  be a checkpoint between $x$ and $s_1$, and between $y$ and $s_1$. In
  particular, it would be a checkpoint in $G$, contradicting the fact
  that $G$ is hard. Hence $z'$ is inner in $H$. By
  Lemma~\ref{lem:glt:reaches} $z'$ is an anchor in $(G,z)$, and hence in
  $H$. Up-to permutation of $z$ and $z'$ we can always assume that $z$
  is minimal.

  Look at the series decomposition of $H$ at $z$. The series factor
  must be reduced to a parallel composition of edges by
  Proposition~\ref{prop:fp:atomic}. We prove that no such edge can
  exist. By contradiction, assume there is one. It provides an
  $x$-$s_2$ path $Q$ in $H$ that contains neither $x'$ nor
  $y'$. By Lemma~\ref{lem:sep_pair_checkpair}, let $P$ be a
  $s_1$-$s_2$ path containing $x$ but not $y$. The path $PxQ$ is from
  $s_1$ to $s_2$, and contains neither $x'$ nor $y'$, a contradiction
  to Lemma~\ref{lem:sep_pair_checkpair}.

  Naming $H_i$ the series arguments, $G$ has the following shape:
  \begin{center}
    \includegraphics{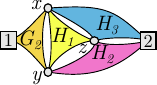}
  \end{center}

  Consider the case when $z\in\{x',y'\}$. Without loss of generality,
  suppose that $z=y'$. If $x'$ is in $H_1$, then we easily prove that
  there are no $x$-$s_2$ inner paths in $H_3$ or $y$-$s_2$ inner paths
  in $H_2$. By contradiction assume, for example, $Q$ to be a $x$-$s_2$
  inner path in $H_3$. By Lemma~\ref{lem:sep_pair_checkpair}, let
  $Q$ be a $s_1$-$s_2$ path in $G$ containing $x$ but not $y$.
  The path $PxQ$ is from $s_1$ to $s_2$, but contains neither $x'$
  nor $y'$, contradicting Lemma~\ref{lem:sep_pair_checkpair}.
  Hence $G$ has the following shape:
  \begin{center}
    \includegraphics{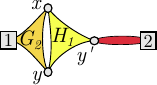}
  \end{center}
  There, $y'$ is a checkpoint, which contradicts $G$ being hard.

  Still assuming $z=y'$, we now assume that $x'$ is in $H_3$. Up-to permuting
  $x$ and $y$, and by extension, $H_2$ and $H_3$,
  this takes care of the case with $x'$ in $H_2$. Using
  Lemma~\ref{lem:sep_pair_checkpair} and similar reasoning as above,
  $x'$ must be a checkpoint between $x$ and $s_2$ in $H_3$, giving the
  shape of \zcase3.

  Thus we may assume $z\not\in\{x',y'\}$, and we do a case analysis
  depending on which $H_i$ contain $x'$ and/or $y'$.

  Assume that $x'$ and $y'$ are both in $H_1$. By
  Lemma~\ref{lem:sep_pair_checkpair} applied to the separation pair
  $(x,y)$, any $x$-$s_2$ inner
  path in $H_3$ or $H_2$ gives rise to a $s_1$-$s_2$ path containing
  neither $x'$ nor $y'$, a contradiction to the latter lemma applied
  to $(x',y')$; $G$ has the following shape:
  \begin{center}
    \includegraphics{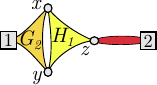}
  \end{center}
  As before, $z$ is a checkpoint in $G$ which is not possible as $G$ is
  hard.

  Assume that $x'$ and $y'$ are both in $H_3$. By symmetry this will
  take care of the case with $x'$ and $y'$ both being in $H_2$. By
  Lemma~\ref{lem:sep_pair_checkpair} again there cannot be $y$-$s_2$
  inner paths in $H_2$: we get the shape of \scase{IH} with $z$
  being indeed minimal.

  Assume $x'$ is in $H_3$ and $y'$ in $H_1$. By symmetry the cases of
  $x'$ in $H_2$ and $y'$ in $H_1$ and/or of permuting $x'$ and $y'$ are
  taken care of by the following too. Again by
  Lemma~\ref{lem:sep_pair_checkpair}, $y$-$s_2$ inner paths do not
  exists in $H_2$, and, without loss of generality, $H_2$ can be
  assumed empty. Similar reasoning as above using
  Lemma~\ref{lem:sep_pair_checkpair} can be used to prove the
  following: $x'$ is a checkpoint between $x$ and $z$, and
  between $x$ and $s_2$ in $H_3$, and $y'$ is a checkpoint between $y$
  and $z$ and between $x$ and $z$ in $H_1$. This is an instance of \zcase3.

  \medskip

  Finally, assume $x'$ and $y'$ are in $H_3$ and $H_2$ respectively.
  Using Lemma~\ref{lem:sep_pair_checkpair}, we prove, following the
  same reasoning as above, that $x'$ is a checkpoint between $x$ and
  $s_2$ in $H_3$ and $y'$ between $y$ and $s_2$ in $H_2$. Said
  otherwise, $G$ has the following shape:
  \begin{center}
    \includegraphics{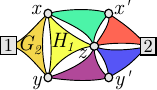}
  \end{center}

  Name $L$ this shape. The graph $(L,x')$ has treewidth $4$. As has
  been done above, we can use Propositions~\ref{prop:subst_minor}
  and~\ref{prop:triangle_property} repeatedly to refine the shapes
  of the subgraphs of $G$ corresponding to the $3$ edges of $L$. In
  all possible cases, retaining minimality of $x$, $x'$, $y$, and
  $y'$, as well as the fact $G$ is hard gives the shape of \scase D.
\end{proof}
We finally deduce that hard terms reach all their minimal separation forget pairs.
\begin{lemma}
  \label{lem:sep:jump}
  Let $G$ be a hard graph with minimal separation forget pair $(x',y')$.
  All parsings $t$ of $G$ reaching a minimal separation pair $(x,y)$ also reach $(x',y')$.
\end{lemma}
\begin{proof}
  We fix $G,x',y'$, we proceed by induction on the size of the component of $(G,x,y)$
  containing $x'$, and we use Proposition~\ref{prop:two-sep-pairs}.
  If $(x,y')$ is a separation pair (\scase Z), then we use Proposition~\ref{lem:glt:reaches} twice: $t$ reaches $x$, thus $(x,y')$, thus $y'$ in particular, thus $(x',y')$. 
  In \scase{IH} we proceed by induction on $(x,z)$, which is reached by $t$ by similar reasoning, and for which the induction measure decreases.
  In \scase{X} and \scase{D}, we conclude directly by \axFX and \axFD, respectively.
  (In those latter two cases, we first need to rewrite $t$ so that it agrees syntactically with the shape of the axiom we use. We use Lemma~\ref{lem:anchor:parsing} for that, which requires showing that some of the inner vertices are anchors in appropriate subgraphs. To do so, we exploit the fact that $G$ is hard, so that the subgraphs corresponding to the edges of the exhibited shape must be connected enough.)
\end{proof}
Together with Lemma~\ref{lem:reaching:sep}, Lemma~\ref{lem:fp:sync} follows, which concludes our completeness proof.

\section{Conclusion and future work}
\label{sec:ccl}

We have provided a finite presentation of graphs of treewidth at most three.
Our axioma\-tisation comprises axioms for dealing with full prime decompositions (\ax{1-7}), axioms for reaching anchors (\axFS, \axFK), and two axioms for hard graphs (\axFX and \axFD).

\medskip

An obvious question for future work is whether the approach presented
here generalises to the case of graphs of treewidth at most $k$.

It does so up to Section~\ref{sec:anchors}.
The syntax we use readily characterises those graphs
(Proposition~\ref{prop:tw:terms}), and the seven first axioms as well
as the results in Section~\ref{sec:completeness} are not specific to
the case $k=3$.
Axioms for anchors (Figure~\ref{fig:anchor:axioms}) also generalise.
Accordingly, our results about series decomposition and anchors
(Lemma~\ref{lem:anchor:parsing}) extend easily. In particular,
non-atomic full prime graphs of arity $0,1$ and $k$ would always have
an anchor.

This means we also have finite presentations for bounds $k=1,2$. (For
$k=2$ this axiomatisation differs from the one previously
proposed~\cite{DBLP:conf/mfcs/Cosme-LlopezP17,DBLP:conf/mfcs/DoczkalP18},
because the chosen syntax is different: there, graphs of arity three
are not considered, and parallel composition may be applied to graphs
of distinct arities.)

However, the results from Section~\ref{sec:hard} are specific to the
case $k=3$, and getting finite presentations for larger values of $k$
seems difficult. Following the strategy presented here, we would need
to prove the forget point agreement lemma (Lemma~\ref{lem:fp:sync})
for new hard graphs, which may have arities between $2$ and $k-1$.

Still, we would like to emphasise the utility of such a generalisation
in the context of graph theory. Robertson and Seymour proved
\cite{DBLP:journals/jct/RobertsonS04} that for all integers $k$ there
must be a finite list of excluded minors characterising the class of
treewidth at most $k$ graphs. Consider a minimal one amongst them, say
$H$, along with a maximal clique over vertices $x_1,\dots,x_i$. It is
not difficult to prove that either $H$ is $K_{k+2}$, or
$(H,x_1,\dots,x_i)$ is hard. Thus, studying hard graphs of treewidth
$k+1$ might lead to a better understanding of excluded minors for
treewidth at most $k$.

\clearpage
\bibliography{short}

\end{document}